\documentclass[a4paper,fleqn]{cas-sc}

\usepackage[authoryear,longnamesfirst]{natbib}
\usepackage{algorithm}
\usepackage{algpseudocode}
\usepackage{booktabs}
\usepackage{multirow}
\usepackage{subcaption}
\usepackage{placeins}
\usepackage{amsthm}
\usepackage{caption}
\usepackage{bm}
\usepackage{graphicx}
\newtheorem{theorem}{Theorem}[section]
\newtheorem{proposition}[theorem]{Proposition}

\theoremstyle{remark}

\begin{document}
\let\WriteBookmarks\relax
\def\floatpagepagefraction{1}
\def\textpagefraction{.001}
\shorttitle{Scalable Bayesian Quantile Regression via EM-INLA}
\shortauthors{T. Rodriguez \& J. Cardona}

\title [mode = title]{Scalable Bayesian Quantile Regression via EM and INLA: The EM-INLA Algorithm}

\author[1]{Tomas Rodriguez Taborda}
\cormark[1]
\ead{torodriguezt@unal.edu.co}
\credit{Conceptualization, Methodology, Software, Writing -- Original Draft}

\author[1]{Johnatan Cardona-Jiménez}
\ead{jcardonj@unal.edu.co}
\credit{Conceptualization, Supervision, Writing -- Review \& Editing}

\affiliation[1]{organization={Statistics Department, Universidad Nacional de Colombia},
                city={Medell\'in},
                country={Colombia}}

\cortext[cor1]{Corresponding author}

\begin{abstract}
We propose EM-INLA, a scalable algorithm for empirical-Bayes hierarchical quantile regression that combines the Expectation--Maximization (EM) algorithm with Integrated Nested Laplace Approximations (INLA). The method exploits the normal--exponential mixture representation of the Asymmetric Laplace Distribution (ALD) to reformulate each M-step as a weighted Gaussian regression handled by INLA, with the ALD scale and random-effect variances updated in closed form at each iteration.
The scale and variance hyperparameters are estimated by marginal maximum likelihood via the EM algorithm, and posterior marginals for the regression and random-effect parameters are obtained from a final INLA call conditional on the converged hyperparameter estimates. An analytic cluster-robust sandwich correction, computed in a single pass from the converged fit, then yields calibrated confidence
intervals for the fixed effects, with coverage verified in simulations against a cluster-bootstrap benchmark.
The result is an algorithm that avoids all MCMC sampling while scaling to large datasets and complex hierarchical structures that are intractable for standard fully Bayesian approaches. In a simulation study across four error scenarios, EM-INLA recovers the true parameters with accuracy comparable to that of Hamiltonian Monte Carlo (HMC) while achieving speedups ranging from $15\times$ to over $50\times$, depending on the scenario. The method is successfully applied to the 2023 Prueba Saber 11 standardized test in Colombia to model the conditional quantiles of student scores as a function of socioeconomic covariates under a hierarchical structure of students nested within schools and municipalities, with over $400{,}000$ observations and $13{,}801$ schools and $1{,}116$ municipalities.
\end{abstract}


\maketitle

\section{Introduction}

Traditional regression models are commonly built by relating a linear predictor to
the mean of the outcome variable via a link function. In contrast, quantile regression
allows one to go beyond the expected value, describing how predictors affect the
response at different locations of its distribution, including the tails.
Unlike classical linear regression fitted via
ordinary least squares, quantile regression employs an asymmetric loss function in
which overestimation below the median is penalized more severely than underestimation,
and conversely for quantiles above the median. This approach is also robust to heteroscedasticity and departures from normality. Beyond its role as a robust
alternative to ordinary least squares, quantile regression provides a complete
characterization of the conditional distribution of the response variable by estimating
an entire family of conditional quantile functions $Q_\tau(Y \mid \mathbf{x})$ for
$\tau \in (0,1)$. This allows one to reveal whether a covariate shifts the entire
distribution uniformly, concentrates its effect on the tails, or leaves certain regions
of the distribution unaffected, patterns that mean regression cannot detect. Under hierarchical data structures, the conditional distribution may vary
across groups not only in its mean but also in its shape, skewness, and tail behavior,
making quantile regression a natural tool for uncovering how covariate effects differ
across the response distribution within and between groups.

The idea of conditional quantiles in the context of classical linear models was initially
proposed by \cite{koenker1978regression}, who showed that sample quantiles in a
location model generalize naturally to the linear model by minimizing an asymmetrically
weighted sum of absolute residuals, giving rise to a new class of estimators termed
regression quantiles. In the following decades, the frequentist theory expanded rapidly:
\cite{powell1986censored} extended the framework to censored outcomes;
\cite{koenker1994quantile} introduced quantile smoothing splines for nonparametric
estimation; and \cite{koenker2004quantile} proposed a penalized likelihood approach for
longitudinal and panel data with subject-specific fixed effects, treating individual
effects as pure location shifts common to all quantiles.

Building upon the frequentist framework, \cite{YU2001437} introduced the Asymmetric
Laplace Distribution (ALD) as a working likelihood for Bayesian quantile regression,
enabling prior-to-posterior updating on the regression parameters and the estimation of
any conditional quantile. Casting these models within a Bayesian framework is valuable because,
by treating parameters as random variables, it provides a natural quantification of
uncertainty for all model parameters and facilitates the incorporation of prior knowledge
through prior distributions. The theoretical justification for the ALD as a working likelihood
was later formalized by \cite{sriram2013posterior}, who established posterior consistency
of the regression parameters even under ALD misspecification, and by
\cite{yang2016posterior}, who derived conditions for asymptotically valid posterior
inference. The representation of the ALD as a mixture of a normal and an
exponential distribution, introduced by \cite{kozumi2011gibbs}, enables the likelihood
to be rewritten hierarchically by conditioning on a latent variable $v_i \sim
\mathrm{Exp}(\sigma^{-1})$, thereby opening the door to inference methods designed for
latent Gaussian models. For longitudinal and clustered settings, \cite{geraci2007quantile}
incorporated subject-specific random intercepts via the ALD likelihood, and
\cite{geraci2014linear} subsequently generalized this to linear quantile mixed models
accommodating multiple random effects and complex dependence structures.

The central goal of Bayesian inference is to obtain the posterior distribution; however, that distribution rarely admits a closed form in more complex models, and so Markov Chain Monte Carlo (MCMC) methods have long served as the standard tool for posterior approximation. Nevertheless,
in the presence of complex hierarchical structures and large datasets, the high
dimensionality of the problem renders MCMC computationally infeasible. As shown by
\cite{10.1214/08-AOS634}, the computational complexity of MCMC estimators based on
Metropolis random walks scales as $O(d^2)$ in the parameter dimension even under
favorable conditions, rendering these methods impractical for hierarchical models with
thousands of random effects. The introduction of Integrated Nested Laplace
Approximations (INLA) by \cite{rue2009approximate} was a major advancement in terms of
computational efficiency, providing an algorithm that scales to large datasets and
high-dimensional parameter spaces by exploiting the sparse precision
structure of Gaussian Markov Random Fields (GMRFs)
\citep{Rue2005}. However, INLA is restricted to latent Gaussian
models, so it cannot be applied to the ALD directly. The normal--exponential mixture
representation of the ALD addresses this limitation, since conditioning on the latent variable
$v_i$ reformulates the model as a latent Gaussian problem, making it compatible with
INLA. The difficulty is that this representation introduces latent variables $v_i$ whose
conditional distribution must be handled at each iteration. \cite{yueyue2011bayesian}
proposed a unified framework combining MCMC and INLA for additive mixed quantile
regression under GMRF priors; however, since the second derivative of the ALD
log-likelihood is zero, their approach required approximating the check function by a
smooth surrogate, introducing a source of error controlled by a tuning parameter whose
choice affects the quality of inference, particularly for extreme quantiles. On the
frequentist side, \cite{bar2021quantile} showed that quantile regression estimating
equations can be solved via an EM algorithm whose M-step reduces to weighted least
squares, with E-step weights computed as expectations of generalized Inverse Gaussian
(GIG) variables, and extended this approach to mixed-effects models. To the best of our knowledge, no existing approach fits the ALD entirely through INLA under a Bayesian framework: available methods either smooth the check loss to
force differentiability or resort to Monte Carlo sampling of the latent variables $v_i$.

Our contribution is to exploit the closed-form GIG conditional moments
of $v_i$, available from the mixture representation of
\cite{kozumi2011gibbs}, to embed INLA within an exact EM algorithm,
reformulating hierarchical quantile regression as a sequence of weighted
Gaussian regressions solved by INLA. We further derive analytical M-step updates for
the ALD scale parameter $\sigma$ and the random-effect variances within
this empirical-Bayes scheme, so that every step of the algorithm is
available in closed form, requiring neither likelihood smoothing nor
latent-variable sampling. The scale and variance hyperparameters are
estimated by marginal maximum likelihood via the EM iterations, and posterior marginals
for the regression and random-effect parameters are obtained from a final INLA call
conditional on the converged hyperparameter estimates, placing the method within an
empirical-Bayes framework. To quantify the additional uncertainty induced by
conditioning on estimated hyperparameters, we complement the point estimates with an
analytic cluster-robust sandwich correction that extends the corrections typically
used for quantile regression under the misspecified ALD. It is computed in a single pass
from the converged fit, and we verify its calibration in simulation against a
cluster bootstrap benchmark. The resulting algorithm, which we call EM-INLA, avoids all MCMC
sampling while scaling to large hierarchical datasets, and delivers both point
estimates and valid interval estimates from one run. In a simulation study, we show that the EM-INLA algorithm
recovers the true parameters while being substantially faster than the Hamiltonian Monte Carlo (HMC) algorithm; moreover, HMC
fails to converge within a practical time budget for the full hierarchical model, whereas
EM-INLA converges in minutes.

The remainder of the article is organized as follows. Section~\ref{sec:background}
introduces the Bayesian quantile regression framework under the ALD mixture representation. Section~\ref{sec:aline} presents the EM-INLA algorithm, including the closed-form E-step moments and the analytical M-step updates for the scale and variance parameters. Section~\ref{sec:examples} reports a simulation study evaluating parameter recovery and computational cost. In Section~\ref{sec:aplic.data}, an application with a large data set from the Prueba Saber~11, a standardized test in Colombia, is presented. Finally, Section \ref{sec:conclusions} presents some concluding remarks.

\section{Bayesian Quantile Regression}
\label{sec:background}

Bayesian quantile regression was proposed by \cite{YU2001437}, who
introduced the Asymmetric Laplace Distribution (ALD) as a working
likelihood for inference on the conditional quantiles of a response
variable. The ALD is built around the check loss function
\begin{equation}\label{eq:check_loss}
    \rho_\tau(u) = u\bigl(\tau - \mathbf{1}_{(u < 0)}\bigr),
    \quad \tau \in (0,1),
\end{equation}
which introduces an asymmetric penalization where for $\tau > 0.5$,
underestimation is penalized more heavily than overestimation, and
conversely for $\tau < 0.5$.

The corresponding ALD density, with location $\mu$, scale $\sigma > 0$,
and asymmetry parameter $\tau$, is given by
\begin{equation}\label{eq:ald}
    f_{\text{ALD}}(y \mid \mu, \sigma, \tau)
    = \frac{\tau(1-\tau)}{\sigma}
      \exp\!\left\{-\rho_\tau\!\left(\frac{y - \mu}{\sigma}\right)\right\}.
\end{equation}

Thus, the quantile regression model takes the form
\begin{equation}\label{eq:bqr}
    y_i = \mathbf{x}_i^\top \boldsymbol{\beta}_\tau + \varepsilon_i,
    \qquad \varepsilon_i \sim \mathrm{ALD}(0,\sigma,\tau),
    \qquad i=1,\dots,n,
\end{equation}
where $\mathbf{x}_i \in \mathbb{R}^p$ is the covariate vector,
$\boldsymbol{\beta}_\tau\in \mathbb{R}^p$  collects the quantile-specific regression
coefficients, and $\sigma > 0$ is a scale parameter that controls the
dispersion of the errors but does not affect the location of the
conditional quantile. Under model~\eqref{eq:bqr}, the $\tau$-th
conditional quantile of $y_i$ given $\mathbf{x}_i$ is
\begin{equation}\label{eq:quantile_interpretation}
    Q_\tau(y_i \mid \mathbf{x}_i) = \mathbf{x}_i^\top \boldsymbol{\beta}_\tau,
\end{equation}
so that, for $j=1,\ldots,p$, $\beta_{\tau,j}$ measures the marginal effect of $x_{ij}$ on
the $\tau$-th conditional quantile of the response. Thus, the likelihood function related to model (\ref{eq:bqr}) is given by

\begin{equation}\label{eq:likelihood}
    f(y_1,\ldots,y_n|\boldsymbol{\beta}_\tau, \sigma)=\prod_{i=1}^{n}\frac{\tau(1-\tau)}{\sigma}
      \exp\!\left\{-\rho_\tau\!\left(\frac{y_i - \mathbf{x}_i^\top \boldsymbol{\beta}_\tau}{\sigma}\right)\right\}.
\end{equation}

The use of likelihood (\ref{eq:likelihood}) is justified because under the ALD, maximizing the likelihood (\ref{eq:likelihood}) is equivalent to minimizing the loss function (\ref{eq:check_loss}). However, working directly with the ALD can be challenging under a Bayesian approach. There is no prior for $\boldsymbol{\beta}_\tau$ that allows for obtaining a closed-form posterior. And when using approximation algorithms such as MCMC, reaching convergence is hard for problems with large datasets and hierarchical structures because of the large number of parameters. One path to overcoming this difficulty is to adopt an augmented-variable approach. For instance, \cite{kozumi2011gibbs} show that the ALD can be written as a normal--exponential scale mixture representation. If
$Y \sim \mathrm{ALD}(\mu,\sigma,\tau)$, then
\begin{equation}\label{eq:mixture}
    Y = \mu + \theta\, V + \kappa \sqrt{\sigma V}\, Z,
    \qquad
    V \sim \mathrm{Exp}(1/\sigma),
    \quad Z \sim N(0,1),
    \quad V \perp Z,
\end{equation}
with mixture constants
\begin{equation}\label{eq:mixture_const}
    \theta = \frac{1-2\tau}{\tau(1-\tau)},
    \qquad
    \kappa^{2} = \frac{2}{\tau(1-\tau)}.
\end{equation}
Integrating out the latent structure variable $Z$ in the mixture (\ref{eq:mixture}) yields the following conditional distribution
\begin{equation}\label{eq:cond_gauss}
    Y \mid V = \nu \;\sim\; N\!\left(\mu + \theta \nu,\; \kappa^{2}\sigma \nu\right).
\end{equation}

This representation shows that the model is conditionally Gaussian given the latent variables $v_i$. In this work, we use the representation in (\ref{eq:cond_gauss}) to enable fast inference by combining INLA with an EM algorithm \citep{dempster1977maximum}.

\subsection{Hierarchical structure and random effects}

In many applications, observations are naturally grouped into clusters 
such that units within the same group share common characteristics that systematically shift the mean response. Ignoring this structure conflates between-group and within-group variability,  leading to biased estimates of the regression coefficients; therefore model~\eqref{eq:bqr} is extended to incorporate group-level  random effects.

Let there be $K$ blocks of random effects, and denote by $g_k(i) \in \{1,\dots,J_k\}$ 
the level of block $k$ to which observation $i$ belongs, where $J_k$ is the number 
of levels in block $k$. The full linear predictor is
\begin{equation}\label{eq:linear_predictor}
    \eta_i = \mathbf{x}_i^\top \boldsymbol{\beta}_\tau 
           + \sum_{k=1}^{K} \alpha_{g_k(i)}^{(k)},
\end{equation}
where $\boldsymbol{\alpha}^{(k)} = (\alpha_1^{(k)},\dots,\alpha_{J_k}^{(k)})^\top$ 
is the $k$-th block of random effects with exchangeable Gaussian prior 
$\alpha_j^{(k)} \overset{\mathrm{iid}}{\sim} N(0,\sigma_k^2)$, where 
$\sigma_k^2 > 0$ quantifies the between-group variability in block $k$.

Fitting a hierarchical model using INLA requires that the model belongs to the class of latent Gaussian Models, which approximates the posterior marginals of 
the parameters in a more computationally efficient manner, especially compared to MCMC. However, the mixture representation in (\ref{eq:cond_gauss}) of the ALD introduces latent variables $v_i$ whose conditional distribution depends on the unknown parameters  $(\boldsymbol{\beta}_\tau, \sigma, \{\sigma_k^2\})$, 
preventing direct application of INLA. If these variables were 
observed, model~\eqref{eq:bqr} would reduce to a standard latent 
Gaussian model, and INLA could be applied without modification. We resolve this by treating the $v_i$ as missing data and embedding INLA within an Expectation--Maximization (EM) algorithm.

\section{The EM-INLA Algorithm}
\label{sec:aline}
The defining feature of EM-INLA is that every quantity required by
each step is available in closed form, so the algorithm involves
neither latent-variable sampling nor likelihood smoothing. The E-step
moments follow from standard properties of the GIG distribution
\citep{jorgensen1982statistical}; the M-step updates for $\sigma$ and
the random-effect variances are derived in closed form within our
empirical-Bayes scheme.
The algorithm iterates between two steps: the E-step computes the 
conditional moments of the $v_i$ given current parameter estimates, 
replacing the latent variables by their expected values. The M-step maximizes the resulting expected complete-data log-likelihood ($Q$-function) with respect to the parameters. 
The parameter space naturally decomposes into three types of parameters, whose updates are grouped into two strategies:

\begin{itemize}
    \item \textbf{Location parameters} $(\boldsymbol{\beta}_\tau, \{\boldsymbol{\alpha}^{(k)}\})$: updated via a weighted Gaussian regression embedded within the INLA framework.
    \item \textbf{Scale and variance parameters} ($\sigma$ and $\{\sigma_k^2\}$): updated analytically via explicit closed-form expressions.
\end{itemize}

All derivations are presented in Appendix~\ref{app:proofs}.

\subsection{E-step: Conditional moments of the latent variables}
\label{sec:estep}

The E-step computes the two conditional moments of $v_i$ that are
sufficient to evaluate the expected complete-data log-likelihood ($Q$-function). 
Given the current parameter estimates for the location components ($\boldsymbol{\beta}^{(t)}$ and $\{\boldsymbol{\alpha}^{(k)(t)}\}$) and the scale $\sigma^{(t)}$, 
the conditional distribution of each latent variable $v_i$ follows a 
Generalized Inverse Gaussian (GIG) distribution. This arises from combining the 
Gaussian conditional likelihood~\eqref{eq:cond_gauss} with the 
exponential prior on $v_i$, yielding a kernel proportional to 
$v_i^{-1/2}\exp\{-(\chi_i/v_i + \psi v_i)/2\}$. This specific kernel identifies the 
index parameter of the GIG distribution as $\lambda = 1/2$. 

As detailed in Appendix~\ref{app:gig}, the exact conditional distribution is
\begin{equation}\label{eq:gig_conditional}
v_i \mid y_i, \boldsymbol{\beta}^{(t)}, \{\boldsymbol{\alpha}^{(k)(t)}\}, \sigma^{(t)}
\sim \mathrm{GIG}\!\left(\tfrac{1}{2},\, \chi_i,\, \psi\right),
\end{equation}
with parameters
\begin{equation}\label{eq:gig_params}
\chi_i = \frac{r_i^{(t)2}}{\sigma^{(t)}\kappa^2},
\qquad
\psi = \frac{\theta^2}{\sigma^{(t)}\kappa^2} + \frac{2}{\sigma^{(t)}},
\qquad
r_i^{(t)} = y_i - \hat{\mu}_i^{(t)},
\end{equation}
where $\hat{\mu}_i^{(t)} = \mathbf{x}_i^\top\boldsymbol{\beta}^{(t)} + \sum_k \hat{\alpha}_{g_k(i)}^{(k)(t)}$ 
is the current linear predictor. The parameter $\chi_i$ varies across observations, 
driven by the individual-specific residuals $r_i^{(t)}$, while 
$\psi$ remains constant for all observations since the global scale $\sigma$ is 
shared. 

Using the properties of the GIG distribution \citep{jorgensen1982statistical}, the required closed-form moments are evaluated as
\begin{equation}\label{eq:gig_moments}
\mathbb{E}[v_i^{-1} \mid \cdot] 
    = \sqrt{\frac{\psi}{\chi_i}},
\qquad
\mathbb{E}[v_i \mid \cdot] 
    = \sqrt{\frac{\chi_i}{\psi}} + \frac{1}{\psi}.
\end{equation}
These two expectations are all that is required for the M-step: $\mathbb{E}[v_i^{-1} \mid \cdot]$ 
determines the pseudo-responses and precision weights for the location 
update handled by INLA, while $\mathbb{E}[v_i \mid \cdot]$ enters directly into the objective 
function for the closed-form scale update.

\subsection{M-step for \texorpdfstring{$\boldsymbol{\beta}$}{beta} via weighted INLA}
\label{sec:mstep_beta}

Updating the location components $(\boldsymbol{\beta}_\tau, \{\boldsymbol{\alpha}^{(k)}\})$ during the M-step requires maximizing the $Q$-function with respect to these parameters. This maximization then reduces to fitting a weighted Gaussian regression, which is computationally feasible. Theorem~\ref{thm:wls} formalizes this equivalence, with the detailed proof presented in Appendix~\ref{app:wls}.

\begin{theorem}[Weighted Gaussian reformulation of the M-step]
\label{thm:wls}
Define the pseudo-response and precision weights as
\begin{equation}\label{eq:pseudo}
    \tilde{y}_i = y_i - \frac{\theta}{\mathbb{E}[v_i^{-1} \mid \cdot]},
    \qquad
    w_i = \frac{\mathbb{E}[v_i^{-1} \mid \cdot]}{\sigma^{(t)}\kappa^2}.
\end{equation}
Then, the part of the $Q$-function depending on the location parameters is given by
\begin{equation}\label{eq:Q_wls_thm}
    Q(\boldsymbol{\beta}_\tau, \boldsymbol{\alpha})
    = -\frac{1}{2}\sum_{i=1}^{n}
      w_i\!\left(\tilde{y}_i - \mu_i\right)^2 + C,
\end{equation}
where $C$ is a constant independent of $(\boldsymbol{\beta}_\tau, \{\boldsymbol{\alpha}^{(k)}\})$. It follows that maximizing $Q$ is equivalent to fitting the weighted Gaussian model $\tilde{y}_i \sim N(\mu_i,\, w_i^{-1})$.
\end{theorem}

This theorem allows us to update the location parameters using INLA by supplying the pseudo-responses $\tilde{y}_i$ and fixing the observation precisions at $w_i$. The pseudo-response $\tilde{y}_i$ recenters the observation by absorbing the asymmetry of the ALD through $\theta$; when $\tau = 1/2$, $\theta = 0$ and $\tilde{y}_i = y_i$, recovering the standard symmetric case.

\subsection{M-step for the scale parameter \texorpdfstring{$\sigma$}{sigma}}
\label{sec:mstep_sigma}

Because the scale parameter $\sigma$ is not part of the weighted Gaussian formulation used for the fixed and random effects, it must be updated separately outside the INLA framework. By retaining only the terms of the expected complete-data log-likelihood that depend on $\sigma$, we obtain the following objective function (with the full derivation provided in Appendix~\ref{app:sigma_derivs}):
\begin{equation}\label{eq:sigma_objective}
    Q(\sigma) = -\frac{3n}{2}\log\sigma 
    - \frac{1}{\sigma}\sum_{i=1}^{n} A_i,
\end{equation}
where the term $A_i$ is independent of $\sigma$. It is evaluated at the current iteration $t$ by gathering the residuals $r_i^{(t)} = y_i - \hat{\mu}_i^{(t)}$ and the conditional moments from \eqref{eq:gig_moments}:
\begin{equation}\label{eq:A_constant}
    A_i = \frac{1}{\kappa^2}\left[ \frac{1}{2} r_i^{(t)2}\,
    \mathbb{E}[v_i^{-1} \mid \cdot] - \theta\, r_i^{(t)} + 
    \left(\frac{\theta^2}{2} + \kappa^2\right) 
    \mathbb{E}[v_i \mid \cdot] \right].
\end{equation}
The strict positivity of $A_i$ is guaranteed because it corresponds to the conditional expectation $\mathbb{E}[(r_i - \theta v_i)^2/(2\kappa^2 v_i) + v_i \mid \cdot]$; in Appendix~\ref{app:sigma_derivs}, we present the details. Differentiating the objective function~\eqref{eq:sigma_objective} with respect to $\sigma$ and setting it to zero yields the exact closed-form update:
\begin{equation}\label{eq:sigma_homo}
    \hat{\sigma}^{(t+1)} = \frac{2}{3n}\sum_{i=1}^{n} A_i.
\end{equation}
Here, the factor $2/3$ emerges directly from the first-order condition $-\frac{3n}{2} + \frac{1}{\hat{\sigma}}\sum_i A_i = 0$.

\subsection{M-step for the random-effect variances}
\label{sec:mstep_re}

Each variance parameter $\sigma_k^{2}$ contributes to the $Q$-function only through the
Gaussian prior on block $k$, so the $K$ blocks are updated
independently and in parallel. The
$Q$-function is defined as:
\begin{equation}\label{eq:Q_re}
    Q(\sigma_k^2) = -\frac{J_k}{2}\log\sigma_k^2
    - \frac{1}{2\sigma_k^2}\sum_{j=1}^{J_k}
    \mathbb{E}\left[\big(\alpha_j^{(k)}\big)^2 \;\middle|\; 
    \tilde{\mathbf{y}}, \hat{\boldsymbol{\beta}}\right] + C,
\end{equation}
where $C$ is a constant independent of $\sigma_k^2$. Equating the derivative of \eqref{eq:Q_re} with respect to $\sigma_k^2$ to zero, and decomposing the second moment via the standard identity $\mathbb{E}[X^2] = (\mathbb{E}[X])^2 + \mathrm{Var}(X)$, yields the exact closed-form update (see Appendix~\ref{app:re} for the full derivation).

\begin{proposition}[Closed-form update for random-effect variances]
\label{prop:re}
For each block $k = 1, \ldots, K$, the unique maximizer of $Q(\sigma_k^2)$ is given by
\begin{equation}\label{eq:re_update_body}
    \hat{\sigma}_k^{2} = \frac{1}{J_k}\sum_{j=1}^{J_k}
    \left[
      \big(\hat{\alpha}_j^{(k)}\big)^2
      + \mathrm{Var}\big(\alpha_j^{(k)} \mid \tilde{\mathbf{y}}, \hat{\boldsymbol{\beta}}\big)
    \right],
\end{equation}
where $\hat{\alpha}_j^{(k)} = \mathbb{E}\big[\alpha_j^{(k)} \mid \tilde{\mathbf{y}}, \hat{\boldsymbol{\beta}}\big]$ and $\mathrm{Var}\big(\alpha_j^{(k)} \mid \tilde{\mathbf{y}}, \hat{\boldsymbol{\beta}}\big)$ denote the conditional posterior mean and marginal variance of the $j$-th random effect in block $k$, respectively.
\end{proposition}

These two moments do not need to be computed manually; they are extracted directly from the marginal posteriors provided by INLA during the update of the location components.

\subsection{The full algorithm and convergence}
\label{sec:full_algo}

The algorithm is initialized using a preliminary Gaussian fit. Setting the initial scale to $\sigma^{(0)} = \hat{\sigma}_{\text{OLS}}$ provides a stable starting point before introducing the ALD-specific asymmetric weights. The location parameters $\boldsymbol{\beta}^{(0)}$ and $\{\boldsymbol{\alpha}^{(k)(0)}\}$ can simply be initialized at zero, since the first INLA step immediately updates them based on the initial pseudo-responses and weights. Finally, the random-effect variances $\{\sigma_k^{2(0)}\}$ are initialized to a large, diffuse value to prevent excessive early shrinkage of the random effects during the first iteration. Regarding the prior specification, we keep INLA's default priors: for the
fixed effects, each regression coefficient is assigned an independent
weakly-informative Gaussian prior $\beta_j \sim N(0, 10^{3})$, and the
intercept an improper flat prior. These priors are deliberately vague, so
that the posterior is driven by the likelihood. Consistently with the
empirical-Bayes construction, no hyperpriors are placed on the scale or the
random-effect variances. Within each M-step, the Gaussian observation
precisions are fixed to the EM weights $w_i$ and the random-effect precisions
are held fixed at their current estimates $\{\sigma_k^{2(t)}\}$. The same
prior specification is used in both the simulation study and the real-data
application.

Convergence is monitored via the relative absolute change in the scale parameter $\sigma$:
\begin{equation}\label{eq:convergence}
    \frac{|\sigma^{(t)} - \sigma^{(t-1)}|}{\sigma^{(t-1)}} 
    < \texttt{tol}.
\end{equation}
We use $\sigma$ as the primary convergence criterion because it governs the E-step weights $w_i$ and the pseudo-responses $\tilde{y}_i$. Once $\sigma$ stabilizes, the location and variance updates reach a steady state. Algorithm~\ref{alg:em_inla_est} summarizes the complete EM-INLA procedure.

\begin{algorithm}[H]
\caption{EM-INLA: Point Estimation}
\label{alg:em_inla_est}
\begin{algorithmic}[1]

\State \textbf{Initialize.} Set $t = 0$. Obtain 
$\sigma^{(0)}$ from a preliminary Gaussian fit, set $\{\sigma_k^{2(0)}\}$ 
to a diffuse value for $k = 1, \ldots, K$, and initialize 
$\boldsymbol{\beta}^{(0)}$ and $\{\boldsymbol{\alpha}^{(k)(0)}\}$ at zero.

\While{$|\sigma^{(t)} - \sigma^{(t-1)}| / \sigma^{(t-1)} > \texttt{tol}$}

    \State \textbf{E-step.} Compute residuals
    $r_i = y_i - \hat{\mu}_i^{(t)}$, GIG parameters
    \[
        \chi_i = \frac{r_i^2}{\sigma^{(t)}\kappa^2}, \qquad
        \psi = \frac{\theta^2}{\sigma^{(t)}\kappa^2}
               + \frac{2}{\sigma^{(t)}},
    \]
    and conditional moments
    \[
        \mathbb{E}[v_i^{-1} \mid \cdot] = \sqrt{\frac{\psi}{\chi_i}},
        \qquad
        \mathbb{E}[v_i \mid \cdot] = \sqrt{\frac{\chi_i}{\psi}}
                                        + \frac{1}{\psi}.
    \]
    Construct pseudo-response and precision weights
    (Theorem~\ref{thm:wls}):
    \[
        \tilde{y}_i = y_i -
        \frac{\theta}{\mathbb{E}[v_i^{-1} \mid \cdot]},
        \qquad
        w_i = \frac{\mathbb{E}[v_i^{-1} \mid \cdot]}
        {\sigma^{(t)}\kappa^2}.
    \]

    \State \textbf{M-step (INLA).} Fit the weighted Gaussian model
    $\tilde{y}_i \sim N(\mu_i,\, w_i^{-1})$ via INLA with
    hyperparameters fixed at $\{\sigma_k^{2(t)}\}$. Collect posterior
    means $\hat{\boldsymbol{\beta}}^{(t+1)}$,
    $\hat{\alpha}_j^{(k)(t+1)}$ and marginal variances
    $\mathrm{Var}\big(\alpha_j^{(k)} \mid \cdot\big)$ for all blocks $k$.

    \State \textbf{M-step ($\sigma$).} Compute $A_i$ as 
    in~\eqref{eq:A_constant} and apply the closed-form 
    update~\eqref{eq:sigma_homo}:
    \[
        \sigma^{(t+1)} = \frac{2}{3n}\sum_{i=1}^{n} A_i.
    \]

    \State \textbf{M-step (RE variances).} For each block
    $k = 1, \ldots, K$, apply Proposition~\ref{prop:re}:
    \[
        \sigma_k^{2(t+1)} = \frac{1}{J_k}\sum_{j=1}^{J_k}
        \left[
          \big(\hat{\alpha}_j^{(k)}\big)^2
          + \mathrm{Var}\big(\alpha_j^{(k)} \mid \cdot\big)
        \right].
    \]

    \State $t \leftarrow t + 1$.

\EndWhile

\State \textbf{Final fit.} Run INLA once more with the converged 
hyperparameters $\{\hat{\sigma}, \hat{\sigma}_k^2\}$ to 
obtain point estimates and marginal distributions for 
$\boldsymbol{\beta}_\tau$ and all random effects.

\end{algorithmic}
\end{algorithm}

Although Algorithm~\ref{alg:em_inla_est} efficiently provides point estimates for the fixed and random components, the naive credible intervals extracted directly from the final INLA step are typically too narrow to attain any prespecified $100(1-\alpha)\%$ coverage level. This undercoverage has three sources. First, and most fundamentally, the ALD is a working likelihood: by the Bernstein--von Mises theorem for
misspecified models \citep{kleijn2012bernstein, muller2013risk}, the posterior spread converges to the inverse curvature of the working criterion rather than to the sampling variance of the point estimator, so ALD credible sets fail to be confidence sets even asymptotically. Second, the EM algorithm inherently tends to underestimate variances
\citep{louis1982finding}. Third, the final INLA run conditions on the converged EM weights and variance components as if they were fixed and
known, thereby ignoring the uncertainty propagated throughout the EM iterations.

Two remedies for the first problem have been proposed in the ALD literature: analytic sandwich-type corrections, which replace the posterior covariance with a robust estimate built from the quantile
scores \citep{sriram2015sandwich, yang2016posterior}, and resampling \citep{10.1093/jssam/smaf040}. However, neither extends to hierarchical models with random effects, where the estimated cluster effects absorb part of the score variability and covariates may be
constant within clusters. The closest existing approach is the recent Laplace-approximation framework of \citet{nava2026laplace}, which develops a sandwich correction for a single level of random effects without fixed effects.

Below, we derive a closed-form
cluster-robust sandwich estimator for the joint hierarchical
parameterization, computable in a single pass from the converged fit.
As an independent benchmark for its calibration, we also implement a
cluster bootstrap \citep{field2007bootstrapping} that resamples entire
top-level groups and refits the model, preserving the hierarchical
dependence structure (Algorithm~\ref{alg:em_inla_boot}). The simulation
study of Section~\ref{sec:sim_coverage} evaluates both constructions
and shows that they agree closely, at vastly different computational
cost.

\begin{algorithm}
\caption{EM-INLA: Cluster Bootstrap Uncertainty Quantification}
\label{alg:em_inla_boot}
\begin{algorithmic}[1]

\Require Observed data $\mathcal{D} = \{(y_i, \mathbf{x}_i, \mathbf{z}_i)\}_{i=1}^n$
         with a nested hierarchy whose top level (resampling unit) is
         indexed by $h(i) \in \{1,\ldots,M\}$ and whose intermediate
         level is indexed by $s(i)$; quantile level $\tau$; number of
         bootstrap replications $B$; confidence level $1 - \alpha$.

\smallskip
\State \textbf{Point estimates.} Run Algorithm~\ref{alg:em_inla_est} on
$\mathcal{D}$ to obtain
$\hat{\boldsymbol{\beta}}_\tau$, $\hat{\sigma}$, and $\{\hat{\sigma}_k^2\}$.

\smallskip
\For{$b = 1, \ldots, B$}

    \State \textbf{Resample top-level clusters.} Draw $M$ units with
    replacement from $\{1,\ldots,M\}$; let $m_1^*, \ldots, m_M^*$ denote
    the sampled top-level labels (repetitions allowed).

    \State \textbf{Construct bootstrap dataset with fresh nested indices.}
    For each draw $k = 1,\ldots,M$, take \emph{all} observations of the
    sampled top-level cluster $m_k^*$ and assign them the new top-level
    index $k$. Within that draw, relabel its intermediate-level units
    with fresh indices unique to the draw, so that repeated draws of the
    same cluster are treated as independent replicas:
    \[
        \mathcal{D}^{*(b)} = \bigl\{
          \bigl(y_i, \mathbf{x}_i, \mathbf{z}_i,\, k,\, \pi_k(s(i))\bigr)
          : i \in \mathcal{C}_{m_k^*},\; k = 1,\ldots,M
        \bigr\},
    \]
    where $\mathcal{C}_m = \{i : h(i) = m\}$ is the full subtree of
    top-level cluster $m$, and $\pi_k$ is a draw-specific injective
    relabeling of the intermediate-level indices that guarantees
    distinct subgroups across the $M$ draws.

    \State \textbf{Refit.} Run Algorithm~\ref{alg:em_inla_est} on
    $\mathcal{D}^{*(b)}$ to obtain the bootstrap replicate
    $\hat{\boldsymbol{\beta}}_\tau^{*(b)}$.

\EndFor

\smallskip
\State \textbf{Bootstrap distribution.}
Collect the $B$ replicates into the matrix
$\mathbf{B}^* \in \mathbb{R}^{B \times p}$ with rows
$\hat{\boldsymbol{\beta}}_\tau^{*(1)}, \ldots,
 \hat{\boldsymbol{\beta}}_\tau^{*(B)}$.

\State \textbf{Confidence intervals.}
For each component $j = 1, \ldots, p$, compute the
$100(1-\alpha)\%$ percentile interval
\[
    \mathrm{CI}_j^{(1-\alpha)}
    =
    \Bigl[
      \hat{F}_{j}^{-1}\!\bigl(\alpha/2\bigr),\;
      \hat{F}_{j}^{-1}\!\bigl(1 - \alpha/2\bigr)
    \Bigr],
\]
where $\hat{F}_j^{-1}$ denotes the empirical quantile function of the
$j$-th column of $\mathbf{B}^*$.

\smallskip
\Ensure Bootstrap confidence intervals
$\{\mathrm{CI}_j^{(1-\alpha)}\}_{j=1}^p$ and the full
distribution $\mathbf{B}^*$ for subsequent analysis.

\end{algorithmic}
\end{algorithm}

The sandwich correction requires no refitting. Consider the two-block
case ($K = 2$) of the linear predictor~\eqref{eq:linear_predictor}, with
the first-level groups nested within the second-level ones, and collect
the latent vector
$\bm\vartheta = (\bm\beta_\tau^\top, \bm\alpha^{(1)\top},
\bm\alpha^{(2)\top})^\top$ together with the combined design
$\mathbf{A} = [\mathbf{X}\;\mathbf{Z}_1\;\mathbf{Z}_2]$, where
$\mathbf{Z}_k$ is the incidence matrix of block $k$ and
$\mathbf{a}_i^\top$ denotes the $i$-th row of $\mathbf{A}$, so that
$\eta_i = \mathbf{a}_i^\top\bm\vartheta$. In this notation, EM-INLA
maximizes the penalized working criterion
$-\sigma^{-1}\sum_i \rho_\tau(y_i - \mathbf{a}_i^\top\bm\vartheta)
- \tfrac12\bm\vartheta^\top\mathbf{P}\bm\vartheta$, where
$\mathbf{P} = \mathrm{blockdiag}\big(\mathbf{0}_p,\,
\sigma_1^{-2}\mathbf{I}_{J_1},\, \sigma_2^{-2}\mathbf{I}_{J_2}\big)$
collects the random-effect prior precisions. The estimator solves the
joint estimating equation whose score has two parts: the data score
$\sigma^{-1}\sum_i \psi_\tau(y_i - \mathbf{a}_i^\top\bm\vartheta)\,
\mathbf{a}_i$, with $\psi_\tau(u) = \tau - \mathbf{1}\{u<0\}$, and the
prior score $-\mathbf{P}\bm\vartheta$. Linearizing around the
pseudo-true value yields the sandwich
\begin{equation}
\widehat{\mathrm{Var}}(\hat{\bm\vartheta})
  = \hat{\mathbf{H}}^{-1}\big(\hat{\mathbf{M}} + \hat{\mathbf{P}}\big)\hat{\mathbf{H}}^{-1},
\qquad
\hat{\mathbf{H}} = \hat\kappa\, \mathbf{A}^\top\mathbf{A} + \hat{\mathbf{P}},
\qquad
\hat{\mathbf{M}} = \frac{J_2}{J_2-1}\sum_{g=1}^{J_2} \mathbf{s}_g \mathbf{s}_g^\top,
\label{eq:sandwich}
\end{equation}
where $\hat\kappa = \hat f/\hat\sigma$ and $\mathbf{s}_g = \hat\sigma^{-1}\sum_{i:\, g_2(i) = g} \psi_\tau(\hat r_i)\,\mathbf{a}_i$ collects the conditional-residual scores of second-level group $g$. The density $\hat f$ of the residuals at the $\tau$-quantile is estimated with the quantile-spacing method of \citet{hendricks1992hierarchical} and the bandwidth of \citet{bofinger1975estimation}. Intervals for $\hat{\bm\beta}_\tau$ use the corresponding block of \eqref{eq:sandwich} with $t_{J_2-1}$ critical values, since the effective number of independent replicates is the number of top-level groups. The outer matrix $\hat{\mathbf{H}}$ is the curvature of the expected criterion. Since the observed Hessian of the pinball loss vanishes almost everywhere, the density $\hat f$ plays the role of local curvature. The term $\hat{\mathbf{P}}$ inside the central matrix $\hat{\mathbf{M}} + \hat{\mathbf{P}}$ arises as follows: the prior score $-\mathbf{P}\bm\vartheta$ evaluated at the realized random effects has covariance $\mathbf{P}\,\mathrm{Var}(\bm\vartheta)\,\mathbf{P} = \mathbf{P}$, because $\mathrm{Var}(\alpha_j^{(k)}) = \sigma_k^2$, and its cross-covariance with the data score vanishes. The prior precision therefore reappears in the score covariance, carrying the variability of the random effects across replications, a source of uncertainty that any computation conditional on the estimated effects ignores. Under a correctly specified model the information identity gives $\hat{\mathbf{M}} + \hat{\mathbf{P}} \approx \hat{\mathbf{H}}$ and \eqref{eq:sandwich} collapses to the classical unconditional mixed-model covariance; without random effects it reduces to the cluster-robust quantile regression estimator of \citet{parente2016quantile}. All inputs ($\hat f$, $\hat\sigma$, $\hat\sigma_1^2$, $\hat\sigma_2^2$, $\hat r_i$) are available at convergence, and the computation is one pass over the data plus a single sparse Cholesky solve.

The derivation is first-order, and treats the plug-in quantities
$\hat f$, $\hat\sigma$, $\hat\sigma_1^2$ and $\hat\sigma_2^2$ as
consistent for their targets. It also rests on three assumptions.
The random effects are exogenous,
$\mathbb{E}[\alpha_j^{(k)} \mid \mathbf{X}] = 0$. They are Gaussian,
which makes the covariance of the prior score equal to
$\mathbf{P}$; a heavier-tailed distribution of effects would inflate
that term. Finally, $\hat{\mathbf{H}}$ uses the constant-curvature simplification
$f_i \equiv f$, so that a single density estimate serves for all
observations. A formal asymptotic treatment under two-level dependence,
along the lines of \citet{kato2012asymptotics}, lies outside the scope
of this paper. We instead assess the estimator empirically in
Section~\ref{sec:sim_coverage}, where we report its Monte Carlo
coverage against the true parameters across the four scenarios and use
the cluster bootstrap as an independent benchmark.

\section{Simulation study}
\label{sec:examples}
We run a simulation study to evaluate EM-INLA under five questions:

\begin{itemize}
    \item Whether the implementation correctly reproduces established asymmetric-Laplace software.
    \item How accurately it recovers the true coefficients relative to the fully Bayesian benchmark and to the frequentist option.
    \item What its computational cost is relative to that benchmark.
    \item Whether its interval estimates are valid, and how much the sandwich correction improves them, compared to the bootstrap and at what computational cost.
    \item How its runtime scales with sample size.
\end{itemize}

\subsection{Design}
\label{sec:sim_design}

All simulation scenarios share a common nested hierarchical structure with two levels of random effects where observations are nested within first-level groups, which are in turn nested within second-level groups. The linear predictor takes the form
\begin{equation}\label{eq:dgp_predictor}
    y_i = \beta_0 + \beta_1\,x_{1i} + \beta_2\,x_{2i} 
        + u_{g_1(i)} + u_{g_2(i)} + \varepsilon_i,
\end{equation}
where $x_1 \sim U(0,2)$, $x_2 \sim N(0,1)$, $u_{g_1} \sim N(0, \sigma_1^2)$, and $u_{g_2} \sim N(0, \sigma_2^2)$. The true parameter values are fixed at $\beta_0 = 250$, $\beta_1 = 10$, $\beta_2 = -5$, $\sigma_1 = 8$, and $\sigma_2 = 4$. The four scenarios differ solely in the distribution of the error term $\varepsilon_i$ as follows:

\begin{itemize}
    \item[\textbf{M1}] \textbf{Homoscedastic Gaussian.} 
    $\varepsilon_i \sim N(0, 15^2)$. Serves as the baseline, evaluating the methods under a correctly specified homoscedastic setting.

    \item[\textbf{M2}] \textbf{Homoscedastic heavy-tailed.} 
    $\varepsilon_i \sim t_3 \times 3$. Tests the robustness of the ALD working likelihood against severe, heavy-tailed departures from normality.

    \item[\textbf{M3}] \textbf{Heteroscedastic linear scale.} 
    $\varepsilon_i \sim N(0,\; (10(1 + 0.2\,x_{1i}))^2)$. 
    Introduces heteroscedasticity linearly dependent on $x_1$, evaluating EM-INLA under a misspecified variance structure.

    \item[\textbf{M4}] \textbf{Heteroscedastic ALD scale.} 
    $\varepsilon_i \sim \mathrm{ALD}(0,\; \exp(2 + 0.5\,x_{1i}),\; 0.5)$. 
    Features an error scale that varies exponentially with $x_1$, providing a stress test under a heteroscedastic ALD distribution.
\end{itemize}

For models M1 and M2, the error distributions belong to location-scale families centered at zero. The true conditional quantile function for these models is strictly linear, $Q_\tau(y_i \mid \mathbf{x}_i) = \mathbf{x}_i^\top\boldsymbol{\beta}(\tau)$, for all $\tau \in (0,1)$. The slope coefficients remain constant across quantiles ($\beta_1(\tau) = 10$ and $\beta_2(\tau) = -5$), while the intercept shifts according to the quantile level computed as $\beta_0(\tau) = 250 + 15\,\Phi^{-1}(\tau)$ under M1, and $\beta_0(\tau) = 250 + 3\,F_{t_3}^{-1}(\tau)$ under M2, where $\Phi^{-1}$ and $F_{t_3}^{-1}$ denote the standard normal and Student-$t_3$ quantile functions, respectively.

Under the heteroscedastic models (M3 and M4), the error scale depends directly on $x_1$. By design, the $\tau = 0.5$ quantile of $\varepsilon_i$ remains exactly zero for all $i$, guaranteeing that the median conditional quantile is strictly linear with recoverable slopes $\beta_1(0.5) = 10$ and $\beta_2(0.5) = -5$. However, for $\tau \neq 0.5$, the observation-specific scales shift the quantile in a way that cannot be captured by a linear predictor in $\mathbf{x}_i$ alone. As a result, the true conditional quantile functions are non-linear, and the slope coefficients do not admit fixed, closed-form true values away from the median.

Because true coefficient values are not well-defined across all quantiles in M3 and M4, traditional metrics based on parameter recovery, such as the mean squared error of the slopes, are insufficient for a global performance comparison. To evaluate the ability of the proposed algorithm to accurately estimate the conditional quantile function across all scenarios and quantile levels, we report the empirical check loss:
\begin{equation}\label{eq:pinball_eval}
    \mathcal{L}_\tau(\hat{\boldsymbol{\beta}}) 
    = \frac{1}{n}\sum_{i=1}^{n}
      \rho_\tau\!\left(y_i - \mathbf{x}_i^\top\hat{\boldsymbol{\beta}}(\tau)\right).
\end{equation}
This metric allows us to quantify predictive accuracy in a unified manner, independently of whether the underlying true parameters admit a closed-form expression.

\FloatBarrier

\subsection{Validation against existing ALD software}

As a first check, we validate our implementation against two common R packages for frequentist asymmetric Laplace models: \texttt{lqmm} \citep{JSSv057i13} and \texttt{qrLMM} \citep{qrLMM}. While both packages rely on a frequentist methodology, they differ in their estimation algorithms. The \texttt{lqmm} package uses numerical quadrature to approximate the integral over the random effects \citep{JSSv057i13}, whereas \texttt{qrLMM} employs a Stochastic Approximation of the Expectation-Maximization (SAEM) algorithm \citep{galarza2017quantile}. To the best of our knowledge, no comparable R package is available within the Bayesian framework; instead, the standard approach relies on probabilistic programming languages like Stan \citep{carpenter2017stan}, which we employ later as the fully Bayesian benchmark. A primary limitation of these existing frequentist packages is their restriction to a single grouping factor. To enable a direct comparison, we reduced the experimental design from Section~\ref{sec:sim_design} to a single-level random effect structure. Table~\ref{tab:validation_software} reports the empirical check loss and slope recovery across the four scenarios at the five quantile levels $\tau \in \{0.10, 0.25, 0.50, 0.75, 0.90\}$, using a sample size of $n = 600$ distributed across 25 groups, evaluated over $R = 50$ replications. In terms of computational cost, the mean runtime per fit was $9.1$\,s for EM-INLA, below $0.1$\,s for \texttt{lqmm}, and $119$\,s for \texttt{qrLMM}. The quadrature-based \texttt{lqmm} is thus by far the fastest in this single-level setting, while EM-INLA runs roughly an order of magnitude faster than the SAEM-based \texttt{qrLMM}.

\begin{table}[pos=!htbp]
\centering
\footnotesize
\caption{Single-level validation of EM-INLA against the established
ALD mixed-model packages \texttt{lqmm} and \texttt{qrLMM}
($R = 50$ replications, $n = 600$, $J_1 = 25$). Pinball loss is
computed on the fixed-effect predictor; slope RMSEs are reported only
where the true slope is a fixed constant (M1--M2 at all $\tau$;
M3--M4 at $\tau = 0.5$), with dashes (---) elsewhere.}
\label{tab:validation_software}
\setlength{\tabcolsep}{3.5pt}
\begin{tabular}{llccccccccc}
\toprule
 & & \multicolumn{3}{c}{\textbf{EM-INLA}}
   & \multicolumn{3}{c}{\textbf{lqmm}}
   & \multicolumn{3}{c}{\textbf{qrLMM}} \\
\cmidrule(lr){3-5}\cmidrule(lr){6-8}\cmidrule(lr){9-11}
\textbf{Sc.} & $\bm{\tau}$
 & Pinball & RMSE $\beta_1$ & RMSE $\beta_2$
 & Pinball & RMSE $\beta_1$ & RMSE $\beta_2$
 & Pinball & RMSE $\beta_1$ & RMSE $\beta_2$ \\
\midrule
\multirow{5}{*}{M1}
 & 0.10 & 3.051 & 1.936 & 0.868 & 3.032 & 2.034 & 1.217 & 2.971 & 2.366 & 1.311 \\
 & 0.25 & 5.426 & 1.150 & 0.975 & 5.460 & 1.234 & 0.958 & 5.390 & 1.620 & 0.915 \\
 & 0.50 & 6.744 & 1.257 & 0.880 & 6.786 & 1.266 & 0.863 & 6.734 & 1.607 & 0.944 \\
 & 0.75 & 5.391 & 1.416 & 0.862 & 5.419 & 1.558 & 0.887 & 5.356 & 1.770 & 0.888 \\
 & 0.90 & 3.008 & 1.605 & 1.078 & 2.988 & 2.067 & 1.082 & 2.939 & 1.848 & 1.154 \\
\addlinespace
\multirow{5}{*}{M2}
 & 0.10 & 2.327 & 0.447 & 0.314 & 1.876 & 0.740 & 0.394 & 1.668 & 1.289 & 0.759 \\
 & 0.25 & 3.325 & 0.326 & 0.192 & 3.229 & 0.392 & 0.260 & 2.997 & 0.895 & 0.541 \\
 & 0.50 & 3.741 & 0.327 & 0.163 & 3.848 & 0.402 & 0.205 & 3.718 & 0.747 & 0.466 \\
 & 0.75 & 3.267 & 0.382 & 0.211 & 3.112 & 0.478 & 0.247 & 2.895 & 0.638 & 0.543 \\
 & 0.90 & 2.234 & 0.568 & 0.318 & 1.738 & 0.762 & 0.410 & 1.587 & 1.033 & 0.654 \\
\addlinespace
\multirow{5}{*}{M3}
 & 0.10 & 2.642 & ---   & ---   & 2.552 & ---   & ---   & 2.526 & ---   & ---   \\
 & 0.25 & 4.635 & ---   & ---   & 4.681 & ---   & ---   & 4.581 & ---   & ---   \\
 & 0.50 & 5.727 & 0.963 & 0.672 & 5.784 & 1.068 & 0.699 & 5.716 & 1.271 & 0.824 \\
 & 0.75 & 4.590 & ---   & ---   & 4.621 & ---   & ---   & 4.536 & ---   & ---   \\
 & 0.90 & 2.592 & ---   & ---   & 2.495 & ---   & ---   & 2.487 & ---   & ---   \\
\addlinespace
\multirow{5}{*}{M4}
 & 0.10 & 6.710  & ---   & ---   & 8.141  & ---   & ---   & 6.675  & ---   & ---   \\
 & 0.25 & 11.060 & ---   & ---   & 11.043 & ---   & ---   & 11.026 & ---   & ---   \\
 & 0.50 & 13.248 & 2.043 & 1.180 & 13.254 & 2.027 & 1.115 & 13.244 & 2.114 & 1.146 \\
 & 0.75 & 11.133 & ---   & ---   & 11.106 & ---   & ---   & 11.095 & ---   & ---   \\
 & 0.90 & 6.806  & ---   & ---   & 8.232  & ---   & ---   & 6.769  & ---   & ---   \\
\bottomrule
\end{tabular}
\end{table}

The results of EM-INLA are similar to those of \texttt{lqmm}, which exhibits a slightly higher RMSE for both slopes, and both methods clearly outperform \texttt{qrLMM} in RMSE. In terms of runtime, \texttt{lqmm} runs almost instantaneously, whereas EM-INLA requires a few seconds, though it remains substantially faster than \texttt{qrLMM}. The high speed of \texttt{lqmm} is due to its deterministic numerical quadrature, which is limited to a single random effect; EM-INLA instead approximates full Bayesian posterior marginals within a more general hierarchical framework. We consider the extra seconds a reasonable price for accommodating more complex hierarchical models.
\FloatBarrier

\subsection{Comparison with the frequentist benchmark}

We then turn to estimation accuracy relative to established alternatives. The first benchmark is QREM \citep{bar2021quantile}, a frequentist method for quantile regression with two-level random effects and thus one of the few alternatives that allows for multiple grouping factors. This comparison is conducted over a full set of $R = 100$ replications. Table~\ref{tab:precision_QREM} reports the computational runtime, empirical check loss, and both the root mean square error (RMSE) and bias for the slope coefficients and random-effect standard deviations. Results are evaluated at five quantile levels, $\tau \in \{0.10, 0.25, 0.50, 0.75, 0.90\}$, across all four simulation scenarios.

\begin{table}[pos=p]
\centering
\small
\caption{Point-estimation accuracy and computation time of EM-INLA and
         QREM under scenarios M1--M4 ($R = 100$ replications).}
\label{tab:precision_QREM}
\setlength{\tabcolsep}{4pt}
\begin{tabular}{cclrrrrrrrr}
\toprule
\textbf{Model} & \textbf{$\tau$} & \textbf{Method} &
\textbf{t (s)} & \textbf{Pinball} &
\textbf{RMSE$_{\beta_1}$} & \textbf{Bias$_{\beta_1}$} &
\textbf{RMSE$_{\beta_2}$} & \textbf{Bias$_{\beta_2}$} &
\textbf{RMSE$_{\sigma_1}$} & \textbf{RMSE$_{\sigma_2}$} \\
\midrule
\multirow{10}{*}{M1}
 & \multirow{2}{*}{0.10} & EM-INLA & 12.1 & 3.0905 & 1.956 & $-$0.039 & 0.977 & $+$0.127 & 1.914 & 3.203 \\
 &                       & QREM   &  1.7 & 3.0866 & 1.970 & $-$0.067 & 0.994 & $+$0.125 & 2.048 & 2.922 \\[2pt]
 & \multirow{2}{*}{0.25} & EM-INLA &  9.0 & 5.5021 & 1.354 & $+$0.144 & 0.824 & $+$0.011 & 1.953 & 2.827 \\
 &                       & QREM   &  1.7 & 5.4925 & 1.507 & $+$0.139 & 0.840 & $+$0.004 & 2.826 & 2.815 \\[2pt]
 & \multirow{2}{*}{0.50} & EM-INLA &  8.3 & 6.8798 & 1.182 & $+$0.018 & 0.708 & $-$0.099 & 2.064 & 2.758 \\
 &                       & QREM   &  1.6 & 6.8774 & 1.250 & $+$0.068 & 0.734 & $-$0.125 & 4.887 & 3.052 \\[2pt]
 & \multirow{2}{*}{0.75} & EM-INLA &  9.4 & 5.5298 & 1.218 & $-$0.107 & 0.819 & $-$0.079 & 1.842 & 2.847 \\
 &                       & QREM   &  1.7 & 5.5212 & 1.228 & $-$0.091 & 0.852 & $-$0.074 & 2.469 & 2.768 \\[2pt]
 & \multirow{2}{*}{0.90} & EM-INLA & 12.1 & 3.1209 & 1.674 & $-$0.224 & 1.173 & $+$0.027 & 1.861 & 3.235 \\
 &                       & QREM   &  1.7 & 3.1157 & 1.748 & $-$0.228 & 1.149 & $+$0.027 & 2.088 & 3.017 \\
\midrule
\multirow{10}{*}{M2}
 & \multirow{2}{*}{0.10} & EM-INLA & 12.1 & 2.3915 & 0.520 & $+$0.019 & 0.406 & $-$0.031 & 1.453 & 3.107 \\
 &                       & QREM   &  1.6 & 2.3886 & 0.532 & $+$0.021 & 0.423 & $-$0.025 & 1.392 & 2.684 \\[2pt]
 & \multirow{2}{*}{0.25} & EM-INLA &  7.9 & 3.4138 & 0.354 & $+$0.004 & 0.198 & $-$0.026 & 1.481 & 2.662 \\
 &                       & QREM   &  1.3 & 3.3916 & 0.359 & $+$0.015 & 0.204 & $-$0.028 & 1.485 & 2.555 \\[2pt]
 & \multirow{2}{*}{0.50} & EM-INLA &  5.5 & 3.8583 & 0.288 & $-$0.003 & 0.177 & $-$0.003 & 1.449 & 2.297 \\
 &                       & QREM   &  1.2 & 3.8587 & 0.298 & $-$0.007 & 0.177 & $-$0.007 & 1.448 & 2.549 \\[2pt]
 & \multirow{2}{*}{0.75} & EM-INLA &  7.9 & 3.4249 & 0.343 & $-$0.002 & 0.199 & $+$0.005 & 1.476 & 2.631 \\
 &                       & QREM   &  1.3 & 3.4035 & 0.350 & $+$0.002 & 0.201 & $+$0.009 & 1.479 & 2.553 \\[2pt]
 & \multirow{2}{*}{0.90} & EM-INLA & 12.0 & 2.4124 & 0.567 & $-$0.016 & 0.331 & $+$0.037 & 1.495 & 3.054 \\
 &                       & QREM   &  1.6 & 2.4085 & 0.565 & $-$0.014 & 0.344 & $+$0.043 & 1.457 & 2.645 \\
\midrule
\multirow{10}{*}{M3}
 & \multirow{2}{*}{0.10} & EM-INLA & 12.0 & 2.6919 & --- & --- & --- & --- & 1.667 & 3.179 \\
 &                       & QREM   &  1.6 & 2.6900 & --- & --- & --- & --- & 1.507 & 2.894 \\[2pt]
 & \multirow{2}{*}{0.25} & EM-INLA &  9.1 & 4.7232 & --- & --- & --- & --- & 1.729 & 2.842 \\
 &                       & QREM   &  1.6 & 4.7150 & --- & --- & --- & --- & 1.714 & 2.795 \\[2pt]
 & \multirow{2}{*}{0.50} & EM-INLA &  7.6 & 5.8739 & 0.940 & $+$0.038 & 0.562 & $-$0.051 & 1.769 & 2.684 \\
 &                       & QREM   &  1.6 & 5.8741 & 0.974 & $+$0.030 & 0.598 & $-$0.056 & 2.592 & 2.801 \\[2pt]
 & \multirow{2}{*}{0.75} & EM-INLA &  9.2 & 4.7428 & --- & --- & --- & --- & 1.654 & 2.802 \\
 &                       & QREM   &  1.6 & 4.7346 & --- & --- & --- & --- & 1.807 & 2.668 \\[2pt]
 & \multirow{2}{*}{0.90} & EM-INLA & 12.0 & 2.7197 & --- & --- & --- & --- & 1.704 & 3.221 \\
 &                       & QREM   &  1.6 & 2.7173 & --- & --- & --- & --- & 1.683 & 2.976 \\
\midrule
\multirow{10}{*}{M4}
 & \multirow{2}{*}{0.10} & EM-INLA &  9.6 &  6.7533 & --- & --- & --- & --- & 4.850 & 3.554 \\
 &                       & QREM   &  1.6 &  6.7198 & --- & --- & --- & --- & 7.886 & 5.070 \\[2pt]
 & \multirow{2}{*}{0.25} & EM-INLA &  6.1 & 11.1571 & --- & --- & --- & --- & 4.827 & 2.847 \\
 &                       & QREM   &  1.4 & 11.1075 & --- & --- & --- & --- & 7.863 & 4.116 \\[2pt]
 & \multirow{2}{*}{0.50} & EM-INLA &  5.3 & 13.3510 & 2.307 & $+$0.144 & 1.252 & $-$0.089 & 4.853 & 2.524 \\
 &                       & QREM   &  1.1 & 13.3445 & 2.632 & $+$0.154 & 1.335 & $-$0.092 & 7.964 & 4.031 \\[2pt]
 & \multirow{2}{*}{0.75} & EM-INLA &  5.8 & 11.2031 & --- & --- & --- & --- & 4.865 & 2.715 \\
 &                       & QREM   &  1.4 & 11.1567 & --- & --- & --- & --- & 7.970 & 3.943 \\[2pt]
 & \multirow{2}{*}{0.90} & EM-INLA & 10.0 &  6.8321 & --- & --- & --- & --- & 4.995 & 3.516 \\
 &                       & QREM   &  1.6 &  6.7961 & --- & --- & --- & --- & 8.010 & 4.713 \\
\bottomrule
\end{tabular}
\end{table}

The results show that our proposal attains accuracy similar to that of the frequentist benchmark, both in pinball loss and in the RMSE of the slopes and random-effect standard deviations, while keeping a low runtime. Despite the usual concern about the high computational cost of Bayesian analysis, EM-INLA is comparable in time to a frequentist method that does not require computing and sampling from a posterior distribution. An important advantage also emerges in variance estimation. Frequentist mixed models often suffer from variance collapse, yielding boundary estimates of exactly zero for the random-effect variances \citep{chung2013nondegenerate}. We observed this in our simulations, where QREM collapsed at the first level under scenario M4. EM-INLA avoids this boundary singularity because the closed-form update of Proposition~\ref{prop:re} adds the posterior variance of each random effect to its squared posterior mean, ensuring that the variance estimate remains strictly positive at every iteration.

\FloatBarrier

\subsection{Comparison with the fully Bayesian benchmark}

The second benchmark is the fully Bayesian alternative, HMC \citep{hoffman2014no}, implemented via Stan \citep{carpenter2017stan} using its R interface RStan \citep{rstan}. Due to the prohibitive computational cost of HMC, this comparison is conducted in an independent experiment restricted to $R = 25$ replications. Alongside HMC and our proposed EM-INLA, we also include the frequentist method results, as this provides a stable baseline and exposes the computational bottlenecks that traditional MCMC-based Bayesian approaches face in these settings. Because of the reduced number of replications, the Monte Carlo standard errors in this experiment are approximately twice as large as those reported in Table~\ref{tab:precision_QREM}. Tables~\ref{tab:precision_qrem_hmc_M1M2} and~\ref{tab:precision_qrem_hmc_M3M4} detail the results for the homoscedastic (M1--M2) and heteroscedastic (M3--M4) scenarios, respectively.

\begin{table}[pos=htbp]
\centering
\caption{Point-estimation accuracy and computation time of EM-INLA,
         QREM, and HMC under scenarios M1 and M2
         ($R = 25$ replications).}
\label{tab:precision_qrem_hmc_M1M2}
\setlength{\tabcolsep}{4pt}
\resizebox{\textwidth}{!}{%
\begin{tabular}{cclrrrrrrrr}
\toprule
\textbf{Model} & \textbf{$\tau$} & \textbf{Method} &
\textbf{t (s)} & \textbf{Pinball} &
\textbf{RMSE$_{\beta_1}$} & \textbf{Bias$_{\beta_1}$} &
\textbf{RMSE$_{\beta_2}$} & \textbf{Bias$_{\beta_2}$} &
\textbf{RMSE$_{\sigma_1}$} & \textbf{RMSE$_{\sigma_2}$} \\
\midrule
\multirow{15}{*}{M1}
 & \multirow{3}{*}{0.10} & EM-INLA &  12.7 & 3.0656 & 2.212 & $+$0.422 & 0.888 & $+$0.052 & 1.775 & 3.381 \\
 &                       & QREM   &   1.8 & 3.0619 & 2.169 & $+$0.402 & 0.892 & $-$0.031 & 1.533 & 3.070 \\
 &                       & HMC     & 182.2 & 3.0544 & 1.973 & $+$0.417 & 0.797 & $-$0.045 & 1.663 & 2.221 \\[2pt]
 & \multirow{3}{*}{0.25} & EM-INLA &   9.7 & 5.4532 & 1.740 & $+$0.486 & 0.639 & $-$0.053 & 1.673 & 3.060 \\
 &                       & QREM   &   1.8 & 5.4437 & 2.024 & $+$0.457 & 0.615 & $-$0.091 & 3.005 & 3.180 \\
 &                       & HMC     & 153.9 & 5.4452 & 1.645 & $+$0.457 & 0.541 & $-$0.086 & 1.305 & 2.103 \\[2pt]
 & \multirow{3}{*}{0.50} & EM-INLA &   9.4 & 6.8400 & 1.386 & $+$0.195 & 0.612 & $-$0.275 & 1.917 & 2.951 \\
 &                       & QREM   &   1.7 & 6.8373 & 1.578 & $+$0.303 & 0.603 & $-$0.248 & 4.965 & 3.200 \\
 &                       & HMC     & 145.4 & 6.8373 & 1.350 & $+$0.243 & 0.526 & $-$0.205 & 1.456 & 2.163 \\[2pt]
 & \multirow{3}{*}{0.75} & EM-INLA &   9.3 & 5.5064 & 1.330 & $-$0.052 & 0.892 & $-$0.222 & 1.869 & 2.981 \\
 &                       & QREM   &   1.8 & 5.4971 & 1.416 & $+$0.015 & 0.864 & $-$0.204 & 2.204 & 3.153 \\
 &                       & HMC     & 160.3 & 5.4945 & 1.271 & $-$0.070 & 0.776 & $-$0.250 & 1.565 & 2.093 \\[2pt]
 & \multirow{3}{*}{0.90} & EM-INLA &  12.8 & 3.1010 & 1.641 & $-$0.050 & 1.131 & $-$0.206 & 1.529 & 3.266 \\
 &                       & QREM   &   1.9 & 3.0988 & 1.826 & $-$0.078 & 1.139 & $-$0.187 & 1.831 & 3.337 \\
 &                       & HMC     & 195.5 & 3.0863 & 1.673 & $-$0.087 & 1.094 & $-$0.270 & 1.771 & 2.226 \\
\midrule
\multirow{15}{*}{M2}
 & \multirow{3}{*}{0.10} & EM-INLA &  12.4 & 2.3058 & 0.424 & $-$0.010 & 0.367 & $-$0.018 & 1.483 &  3.161 \\
 &                       & QREM   &   1.6 & 2.3066 & 0.452 & $-$0.009 & 0.375 & $-$0.020 & 1.408 &  2.858 \\
 &                       & HMC     & 267.8 & 2.7062 & 0.439 & $-$0.043 & 0.320 & $-$0.021 & 2.603 & 19.647 \\[2pt]
 & \multirow{3}{*}{0.25} & EM-INLA &   8.2 & 3.3061 & 0.322 & $-$0.091 & 0.139 & $-$0.026 & 1.571 &  2.796 \\
 &                       & QREM   &   1.2 & 3.2879 & 0.332 & $-$0.070 & 0.151 & $-$0.035 & 1.577 &  2.713 \\
 &                       & HMC     & 269.3 & 4.4784 & 0.285 & $-$0.054 & 0.139 & $-$0.018 & 1.459 & 16.894 \\[2pt]
 & \multirow{3}{*}{0.50} & EM-INLA &   6.0 & 3.7399 & 0.264 & $-$0.062 & 0.160 & $-$0.022 & 1.496 &  2.437 \\
 &                       & QREM   &   1.2 & 3.7405 & 0.258 & $-$0.053 & 0.171 & $-$0.018 & 1.492 &  2.514 \\
 &                       & HMC     & 276.3 & 6.9351 & 0.258 & $-$0.078 & 0.142 & $-$0.022 & 1.465 & 16.759 \\[2pt]
 & \multirow{3}{*}{0.75} & EM-INLA &   8.2 & 3.3219 & 0.364 & $-$0.073 & 0.140 & $-$0.039 & 1.490 &  2.784 \\
 &                       & QREM   &   1.3 & 3.3014 & 0.355 & $-$0.096 & 0.159 & $-$0.037 & 1.516 &  2.655 \\
 &                       & HMC     & 276.6 & 9.6328 & 0.339 & $-$0.088 & 0.155 & $-$0.037 & 1.456 & 20.751 \\[2pt]
 & \multirow{3}{*}{0.90} & EM-INLA &  12.4 & 2.3462 & 0.543 & $-$0.107 & 0.224 & $-$0.052 & 1.537 &  3.200 \\
 &                       & QREM   &   1.6 & 2.3421 & 0.546 & $-$0.098 & 0.214 & $-$0.033 & 1.499 &  2.614 \\
 &                       & HMC     & 280.4 &13.5126 & 0.533 & $-$0.106 & 0.200 & $-$0.034 & 3.910 & 22.301 \\
\bottomrule
\end{tabular}%
}
\end{table}

Under M1, the pinball losses of the three methods are nearly
identical, showing that EM-INLA reproduces the point-estimation
quality of exact Bayesian computation at roughly $1/20$ to $1/50$ of
its computational cost in this moderately sized design. 
Under M2, by contrast, the pinball loss of HMC deteriorates markedly
relative to the other two methods, reaching $13.51$ at $\tau = 0.90$
against $2.35$ for EM-INLA, and its RMSE for the second-level standard
deviation $\sigma_2$ ranges from $16.8$ to $22.3$ against a true value
of $4$. With only $25$ replications, a handful of divergent fits
dominates the aggregate RMSE. Because the E-step moments and variance updates are available in closed form
for EM-INLA and involve no stochastic exploration of the posterior, EM-INLA
remains stable under M2, matching its own accuracy in scenario M1. We report
the M2 results for HMC as obtained, without allotting the algorithm more
runtime, since the comparison is intended to reflect the computational budget
under which each method would realistically be used.

\begin{table}[pos=htbp]
\centering
\caption{Point-estimation accuracy and computation time of EM-INLA,
         QREM, and HMC under scenarios M3 and M4
         ($R = 25$ replications).}
\label{tab:precision_qrem_hmc_M3M4}
\setlength{\tabcolsep}{4pt}
\resizebox{\textwidth}{!}{%
\begin{tabular}{cclrrrrrrrr}
\toprule
\textbf{Model} & \textbf{$\tau$} & \textbf{Method} &
\textbf{t (s)} & \textbf{Pinball} &
\textbf{RMSE$_{\beta_1}$} & \textbf{Bias$_{\beta_1}$} &
\textbf{RMSE$_{\beta_2}$} & \textbf{Bias$_{\beta_2}$} &
\textbf{RMSE$_{\sigma_1}$} & \textbf{RMSE$_{\sigma_2}$} \\
\midrule
\multirow{15}{*}{M3}
 & \multirow{3}{*}{0.10} & EM-INLA &  12.4 & 2.6652 & --- & --- & --- & --- & 1.564 & 3.403 \\
 &                       & QREM   &   1.6 & 2.6639 & --- & --- & --- & --- & 1.338 & 3.095 \\
 &                       & HMC     & 244.5 & 2.6576 & --- & --- & --- & --- & 1.520 & 2.087 \\[2pt]
 & \multirow{3}{*}{0.25} & EM-INLA &   8.9 & 4.6666 & --- & --- & --- & --- & 1.497 & 3.002 \\
 &                       & QREM   &   1.4 & 4.6573 & --- & --- & --- & --- & 1.334 & 3.025 \\
 &                       & HMC     & 224.6 & 4.6576 & --- & --- & --- & --- & 1.263 & 2.018 \\[2pt]
 & \multirow{3}{*}{0.50} & EM-INLA &   7.6 & 5.8207 & 1.138 & $+$0.246 & 0.460 & $-$0.179 & 1.693 & 2.918 \\
 &                       & QREM   &   1.6 & 5.8209 & 1.154 & $+$0.238 & 0.514 & $-$0.234 & 2.213 & 3.030 \\
 &                       & HMC     & 225.3 & 5.8329 & 1.100 & $+$0.190 & 0.417 & $-$0.151 & 1.403 & 2.253 \\[2pt]
 & \multirow{3}{*}{0.75} & EM-INLA &   9.8 & 4.7049 & --- & --- & --- & --- & 1.653 & 3.081 \\
 &                       & QREM   &   1.6 & 4.6984 & --- & --- & --- & --- & 1.439 & 3.064 \\
 &                       & HMC     & 230.8 & 4.7003 & --- & --- & --- & --- & 1.473 & 1.976 \\[2pt]
 & \multirow{3}{*}{0.90} & EM-INLA &  12.6 & 2.6894 & --- & --- & --- & --- & 1.456 & 3.331 \\
 &                       & QREM   &   1.6 & 2.6919 & --- & --- & --- & --- & 1.475 & 3.275 \\
 &                       & HMC     & 258.2 & 2.6824 & --- & --- & --- & --- & 1.544 & 2.092 \\
\midrule
\multirow{15}{*}{M4}
 & \multirow{3}{*}{0.10} & EM-INLA &  10.1 &  6.7170 & --- & --- & --- & --- & 4.073 & 3.682 \\
 &                       & QREM   &   1.5 &  6.6816 & --- & --- & --- & --- & 8.017 & 4.486 \\
 &                       & HMC     & 203.3 &  6.6892 & --- & --- & --- & --- & 6.032 & 4.037 \\[2pt]
 & \multirow{3}{*}{0.25} & EM-INLA &   5.9 & 11.0186 & --- & --- & --- & --- & 4.713 & 2.750 \\
 &                       & QREM   &   1.4 & 10.9690 & --- & --- & --- & --- & 7.879 & 3.828 \\
 &                       & HMC     & 197.2 & 10.9763 & --- & --- & --- & --- & 2.299 & 2.320 \\[2pt]
 & \multirow{3}{*}{0.50} & EM-INLA &   5.4 & 13.1335 & 2.219 & $+$0.913 & 1.026 & $+$0.100 & 5.369 & 2.675 \\
 &                       & QREM   &   1.1 & 13.1294 & 2.410 & $+$0.717 & 1.007 & $+$0.025 & 8.000 & 4.128 \\
 &                       & HMC     & 197.3 & 13.1365 & 2.292 & $+$0.972 & 1.065 & $+$0.119 & 1.993 & 2.128 \\[2pt]
 & \multirow{3}{*}{0.75} & EM-INLA &   5.4 & 11.0016 & --- & --- & --- & --- & 5.157 & 2.450 \\
 &                       & QREM   &   1.3 & 10.9490 & --- & --- & --- & --- & 8.000 & 3.646 \\
 &                       & HMC     & 199.1 & 10.9575 & --- & --- & --- & --- & 2.340 & 2.135 \\[2pt]
 & \multirow{3}{*}{0.90} & EM-INLA &  10.4 &  6.7049 & --- & --- & --- & --- & 4.975 & 3.067 \\
 &                       & QREM   &   1.5 &  6.6757 & --- & --- & --- & --- & 8.172 & 4.337 \\
 &                       & HMC     & 220.2 &  6.6859 & --- & --- & --- & --- & 5.847 & 2.836 \\
\bottomrule
\end{tabular}%
}
\end{table}

Under M3 and M4, the pinball losses of the three methods are again
nearly identical across all quantile levels. 
These results show that EM-INLA matches the point-estimation quality of the fully Bayesian benchmark at a small fraction of its
computational cost, and that of the frequentist alternative with only a slightly longer runtime.

\FloatBarrier

\subsection{Interval estimation and coverage}
\label{sec:sim_coverage}

We now examine uncertainty quantification. Three interval constructions are compared: (i) the naive empirical-Bayes intervals extracted directly from the final EM-INLA step; (ii) the analytic cluster-robust sandwich correction of Equation~\eqref{eq:sandwich}; and (iii) the cluster bootstrap procedure of Algorithm~\ref{alg:em_inla_boot}. Table~\ref{tab:cobertura} presents the empirical coverage and mean interval length for a nominal $95\%$ confidence level, based on $R = 100$ replications across the four simulation scenarios.

\begin{table}[pos=!htbp]
\centering
\small
\caption{Empirical coverage and mean length of 95\% intervals for regression coefficients ($R = 100$ replications). Compares naive EM-INLA estimates, the analytic cluster-robust sandwich correction of Equation~\eqref{eq:sandwich}, and the cluster bootstrap (Algorithm~\ref{alg:em_inla_boot}). Dashes (---) indicate non-evaluable cases where the true coefficient is not a fixed constant.}
\label{tab:cobertura}
\setlength{\tabcolsep}{4pt}
\begin{tabular}{llccccccccc}
\toprule
 & & \multicolumn{3}{c}{\textbf{Naive (INLA)}}
   & \multicolumn{3}{c}{\textbf{Cluster sandwich}}
   & \multicolumn{3}{c}{\textbf{Cluster bootstrap}} \\
\cmidrule(lr){3-5}\cmidrule(lr){6-8}\cmidrule(lr){9-11}
\textbf{Sc.} & \textbf{$\tau$}
 & $\beta_0$ & $\beta_1$ & $\beta_2$
 & $\beta_0$ & $\beta_1$ & $\beta_2$
 & $\beta_0$ & $\beta_1$ & $\beta_2$ \\
\midrule
\multicolumn{11}{l}{\textit{Panel (a): Empirical coverage (\%)}} \\
\midrule
\multirow{3}{*}{M1}
 & 0.25 & 89 & 30 & 33 & 98 & 97 & 96 & 87 & 96 & 95 \\
 & 0.50 & 89 & 33 & 29 & 99 & 96 & 94 & 84 & 94 & 96 \\
 & 0.75 & 82 & 32 & 30 & 97 & 98 & 92 & 83 & 98 & 94 \\
\addlinespace
\multirow{3}{*}{M2}
 & 0.25 & 94 & 33 & 36 & 99 & 92 & 89 & 82 & 92 & 89 \\
 & 0.50 & 93 & 38 & 40 & 98 & 96 & 96 & 84 & 95 & 93 \\
 & 0.75 & 94 & 30 & 37 & 99 & 95 & 94 & 83 & 95 & 94 \\
\addlinespace
\multirow{3}{*}{M3}
 & 0.25 & --- & --- & --- & --- & --- & --- & --- & --- & --- \\
 & 0.50 & 92 & 37 & 39 & 99 & 93 & 98 & 84 & 90 & 97 \\
 & 0.75 & --- & --- & --- & --- & --- & --- & --- & --- & --- \\
\addlinespace
\multirow{3}{*}{M4}
 & 0.25 & --- & --- & --- & --- & --- & --- & --- & --- & --- \\
 & 0.50 & 72 & 41 & 55 & 95 & 91 & 95 & 84 & 92 & 92 \\
 & 0.75 & --- & --- & --- & --- & --- & --- & --- & --- & --- \\
\midrule
\multicolumn{11}{l}{\textit{Panel (b): Mean interval length}} \\
\midrule
\multirow{3}{*}{M1}
 & 0.25 & 7.99 & 0.86 & 0.49 & 11.10 & 4.44 & 2.54 & 7.44 & 4.22 & 2.47 \\
 & 0.50 & 7.79 & 0.92 & 0.51 & 10.74 & 4.22 & 2.41 & 7.10 & 3.87 & 2.24 \\
 & 0.75 & 7.99 & 0.86 & 0.50 & 11.00 & 4.31 & 2.54 & 7.48 & 4.17 & 2.45 \\
\addlinespace
\multirow{3}{*}{M2}
 & 0.25 & 8.01 & 0.24 & 0.14 & 10.05 & 1.11 & 0.64 & 5.58 & 1.05 & 0.59 \\
 & 0.50 & 8.16 & 0.24 & 0.14 & 10.02 & 0.94 & 0.54 & 5.59 & 0.86 & 0.48 \\
 & 0.75 & 8.06 & 0.24 & 0.14 & 10.08 & 1.10 & 0.66 & 5.73 & 1.06 & 0.62 \\
\addlinespace
\multirow{3}{*}{M3}
 & 0.25 & --- & --- & --- & --- & --- & --- & --- & --- & --- \\
 & 0.50 & 8.03 & 0.72 & 0.42 & 10.52 & 3.33 & 1.88 & 6.33 & 3.08 & 1.81 \\
 & 0.75 & --- & --- & --- & --- & --- & --- & --- & --- & --- \\
\addlinespace
\multirow{3}{*}{M4}
 & 0.25 & --- & --- & --- & --- & --- & --- & --- & --- & --- \\
 & 0.50 & 6.49 & 1.91 & 1.04 & 11.41 & 6.58 & 3.75 & 7.99 & 6.27 & 3.35 \\
 & 0.75 & --- & --- & --- & --- & --- & --- & --- & --- & --- \\
\bottomrule
\end{tabular}
\end{table}

As shown in Panel~(a), the naive approach yields severe undercoverage
for the slope coefficients ($\beta_1$ and $\beta_2$), attaining
empirical coverage rates of only $29\%$ to $55\%$ across all scenarios
and quantile levels. Panel~(b) shows that by conditioning on the
converged EM weights and variance components as if they were known, the
direct EM-INLA fit produces slope intervals that are too narrow.

Both corrections repair this deficiency, and they agree closely with
each other. The sandwich intervals attain coverage between $89\%$ and
$99\%$ across all evaluable cells, which for the slopes is
indistinguishable from the bootstrap (whose coverage ranges from $89\%$
to $98\%$; the Monte Carlo standard error with $R = 100$ is roughly
$2$ to $3$ percentage points). Panel~(b) shows that the two methods
produce slope intervals of nearly identical length. They differ in
cost: the bootstrap refits the model $B = 200$ times, whereas the
sandwich is computed once from quantities already available at
convergence.

The intercept requires a separate comment, since it isolates the role
of the $\hat{\mathbf{P}}$ term in~\eqref{eq:sandwich}. Unlike the
slopes, whose covariates vary within clusters, the intercept absorbs
the cluster-level random effects. The mean of the realized effects is
not identified separately from the intercept and migrates into
$\hat\beta_0$ in every replication, so its sampling variability is dominated by between-cluster variation: its effective sample size is the number of clusters, not the number of observations. A sandwich built only on the conditional data scores treats the estimated random
effects as fixed and is therefore blind to this source of uncertainty. In our experiments, dropping the $\hat{\mathbf{P}}$ term from the central matrix collapsed the intercept coverage to as low as $16\%$ in scenario M2,
where the random-effect variance dominates the residual noise, even while slope coverage remained nominal. With the $\hat{\mathbf{P}}$ term in place, the between-cluster variability is restored and the intercept coverage reaches the $95\%$ to $99\%$ reported in Table~\ref{tab:cobertura}, which is better calibrated than the bootstrap intercept intervals ($82\%$ to $87\%$). The latter inherit the downward bias of the percentile method when only $J_2 = 10$ top-level clusters are resampled. The same mechanism protects coefficients of covariates that are constant or strongly clustered
within groups, whose identification rests on between-cluster
comparisons. This is immaterial in the simulation, where both
covariates vary within clusters, but becomes relevant in the
application of Section~\ref{sec:aplic.data}, where the socioeconomic
covariates are clustered at the school or neighborhood level.

We therefore conclude that the EM-INLA fit supplies the point
estimates, while the analytic cluster sandwich provides valid inference
on the regression coefficients at negligible additional cost. The
cluster bootstrap, run here as an independent benchmark, confirms the
sandwich calibration while requiring $B$ full refits.

\subsection{Computational scaling}
\label{sec:sim_scaling}

To address the final question of the simulation study, we design two
experiments to evaluate the computational scalability of the EM-INLA
algorithm. The first isolates the effect of increasing the number of
random effects at a fixed sample size, while the second assesses
performance as the sample size $n$ and the hierarchical structure grow
proportionally.

In the first experiment, we fix the sample size at $n = 50{,}000$ while
varying the number of first-level groups from $J_1 = 50$ to
$J_1 = 5{,}000$, representing a hundredfold increase in the number of
random effects. Table~\ref{tab:scalability-latent} displays the total
time required to fit all five quantile levels simultaneously.

\begin{table}[pos=!htbp]
\centering
\caption{Total EM-INLA runtime (seconds) to fit five quantile levels at a fixed sample size of $n = 50{,}000$, as the number of random effects grows.}
\label{tab:scalability-latent}
\begin{tabular}{rr cccc}
\toprule
 & & \multicolumn{4}{c}{Total time to fit five quantiles (s), mean\,(sd)} \\
\cmidrule(lr){3-6}
$J_1$ & $J_2$ & M1 & M2 & M3 & M4 \\
\midrule
50   & 4   & 95.8\,(4.2)  & 102.4\,(1.4) & 94.4\,(0.5)  & 93.1\,(2.3)  \\
100  & 8   & 92.6\,(1.3)  & 101.1\,(1.0) & 93.3\,(1.3)  & 87.8\,(2.6)  \\
250  & 21  & 87.7\,(1.9)  & 101.2\,(1.3) & 90.3\,(4.5)  & 87.5\,(1.9)  \\
500  & 42  & 88.0\,(0.7)  & 103.9\,(2.1) & 88.1\,(0.8)  & 88.2\,(4.1)  \\
1000 & 83  & 92.1\,(4.1)  & 109.7\,(1.5) & 92.9\,(3.8)  & 98.2\,(3.7)  \\
2500 & 208 & 104.6\,(7.3) & 91.5\,(1.4)  & 92.3\,(1.2)  & 90.2\,(7.5)  \\
5000 & 417 & 110.5\,(1.7) & 119.4\,(2.2) & 115.5\,(0.7) & 107.8\,(5.6) \\
\bottomrule
\end{tabular}
\end{table}

Table~\ref{tab:scalability-latent} shows that the total computational
time remains stable despite the hundredfold increase in the
latent dimension, growing only marginally from roughly $95$ to $115$
seconds. This indicates that the computational cost of EM-INLA is
governed almost entirely by the sample size $n$ rather than by the
dimension of the random effects, which motivates the second experiment.

In the second experiment, we scale the sample size from $n = 2{,}000$
to $n = 400{,}000$, proportionally expanding the nested hierarchy
($J_1 = \lfloor n/30 \rfloor$, $J_2 = \lfloor J_1/12 \rfloor$).
Table~\ref{tab:scaling_time} reports the per-fit runtimes for every
combination of sample size, scenario and quantile level, and
Figure~\ref{fig:scaling_time} summarizes them graphically. The
results demonstrate that even
while the sample size increases two-hundredfold, the computational
time grows by a factor of less than ten, so the wall-clock time scales
sublinearly with respect to $n$: a log--log regression of runtime on
sample size yields fitted exponents between $0.33$ and $0.42$ across
the four scenarios. As previously noted, HMC requires
over 200 seconds for a sample of only $n = 600$ observations, with
costs escalating rapidly as both the sample size and the latent
dimensions grow, rendering it computationally prohibitive for
large-scale applications like the one presented later in
Section~\ref{sec:aplic.data}.

\begin{table}[pos=!htbp]
\centering
\small
\caption{Per-fit runtime (seconds) of EM-INLA as sample size $n$ and random-effect dimensions $(J_1, J_2)$ scale proportionally, by scenario and quantile level.}
\label{tab:scaling_time}
\setlength{\tabcolsep}{5pt}
\begin{tabular}{rrrlrrrrr}
\toprule
\textbf{$n$} & \textbf{$J_1$} & \textbf{$J_2$} & \textbf{Model} & \textbf{$\tau=0.10$} & \textbf{$\tau=0.25$} & \textbf{$\tau=0.50$} & \textbf{$\tau=0.75$} & \textbf{$\tau=0.90$} \\
\midrule
\multirow{4}{*}{2{,}000} & \multirow{4}{*}{67} & \multirow{4}{*}{6} & M1 & 15.5 & 11.7 & 8.5 & 9.0 & 11.9 \\
 & & & M2 & 11.3 & 11.9 & 5.4 & 11.9 & 10.8 \\
 & & & M3 & 11.4 & 8.9 & 5.6 & 8.7 & 10.9 \\
 & & & M4 & 11.3 & 6.6 & 4.5 & 5.8 & 10.6 \\
\midrule
\multirow{4}{*}{5{,}000} & \multirow{4}{*}{167} & \multirow{4}{*}{14} & M1 & 12.5 & 9.4 & 7.7 & 9.0 & 11.8 \\
 & & & M2 & 11.1 & 12.8 & 5.7 & 12.7 & 10.9 \\
 & & & M3 & 11.1 & 9.2 & 6.7 & 8.8 & 11.6 \\
 & & & M4 & 13.4 & 7.6 & 5.4 & 7.6 & 13.1 \\
\midrule
\multirow{4}{*}{10{,}000} & \multirow{4}{*}{333} & \multirow{4}{*}{28} & M1 & 13.2 & 10.2 & 5.7 & 10.1 & 13.0 \\
 & & & M2 & 12.4 & 14.2 & 6.3 & 14.4 & 12.4 \\
 & & & M3 & 13.3 & 10.3 & 7.1 & 10.1 & 12.9 \\
 & & & M4 & 14.9 & 9.7 & 7.3 & 8.5 & 15.0 \\
\midrule
\multirow{4}{*}{25{,}000} & \multirow{4}{*}{833} & \multirow{4}{*}{69} & M1 & 16.9 & 12.5 & 8.0 & 13.2 & 17.4 \\
 & & & M2 & 17.0 & 20.9 & 8.7 & 18.6 & 17.3 \\
 & & & M3 & 16.4 & 12.6 & 5.9 & 12.3 & 16.8 \\
 & & & M4 & 19.9 & 10.0 & 8.1 & 9.9 & 19.3 \\
\midrule
\multirow{4}{*}{50{,}000} & \multirow{4}{*}{1{,}667} & \multirow{4}{*}{139} & M1 & 25.0 & 18.4 & 11.1 & 19.1 & 23.5 \\
 & & & M2 & 26.8 & 27.5 & 12.0 & 27.5 & 24.1 \\
 & & & M3 & 24.0 & 19.4 & 8.8 & 17.9 & 23.6 \\
 & & & M4 & 28.5 & 17.3 & 16.4 & 16.7 & 28.2 \\
\midrule
\multirow{4}{*}{100{,}000} & \multirow{4}{*}{3{,}333} & \multirow{4}{*}{278} & M1 & 39.4 & 27.2 & 14.8 & 27.9 & 37.9 \\
 & & & M2 & 38.4 & 44.3 & 20.4 & 44.0 & 39.2 \\
 & & & M3 & 37.4 & 29.8 & 15.1 & 29.8 & 37.4 \\
 & & & M4 & 44.3 & 24.5 & 22.8 & 24.0 & 44.2 \\
\midrule
\multirow{4}{*}{200{,}000} & \multirow{4}{*}{6{,}667} & \multirow{4}{*}{556} & M1 & 62.2 & 51.7 & 37.5 & 47.1 & 62.9 \\
 & & & M2 & 66.7 & 72.6 & 35.1 & 74.3 & 65.6 \\
 & & & M3 & 63.5 & 50.4 & 25.7 & 49.5 & 62.4 \\
 & & & M4 & 76.1 & 42.3 & 41.4 & 43.2 & 76.8 \\
\midrule
\multirow{4}{*}{400{,}000} & \multirow{4}{*}{13{,}333} & \multirow{4}{*}{1{,}111} & M1 & 114.6 & 86.7 & 70.7 & 95.8 & 114.3 \\
 & & & M2 & 122.9 & 135.9 & 63.5 & 135.1 & 120.0 \\
 & & & M3 & 118.1 & 89.9 & 51.7 & 92.4 & 114.5 \\
 & & & M4 & 138.2 & 81.4 & 44.5 & 79.6 & 139.9 \\
\bottomrule
\end{tabular}
\end{table}

\begin{figure}[pos=!htbp]
    \centering
    \includegraphics[width=0.55\textwidth]{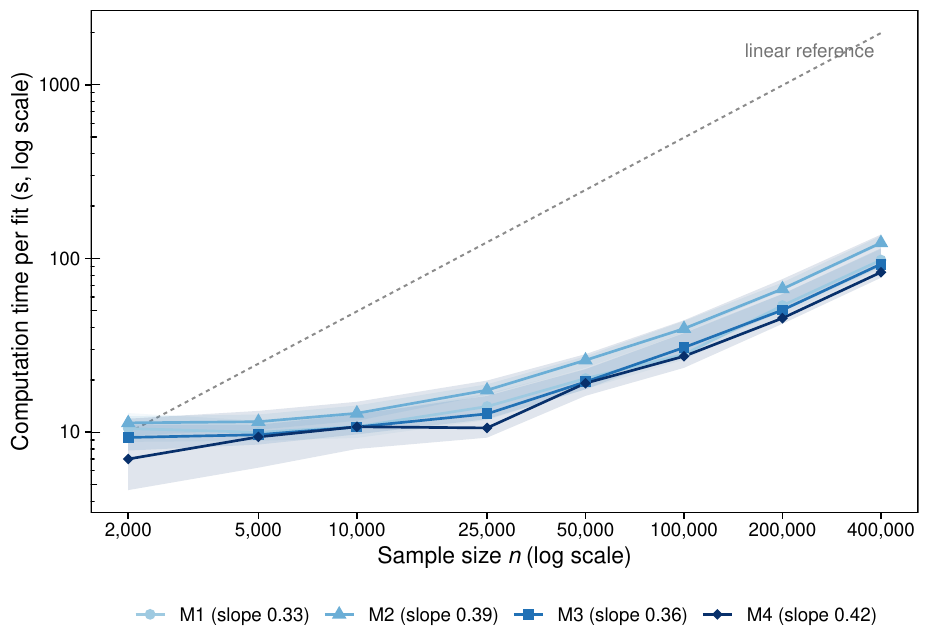}
    \caption{EM-INLA runtime versus sample size $n$ on log--log axes: medians, interquartile bands, and fitted slopes by scenario.}
    \label{fig:scaling_time}
\end{figure}

The two experiments show that the cost of EM-INLA is insensitive to the dimension of the latent field and grows sublinearly with sample size, establishing the scalability of the method. In the next section, one example at the national scale in Colombia is presented.

\FloatBarrier

\section{Application to real data: Prueba Saber 11 in Colombia}
\label{sec:aplic.data}

\subsection{Data and model specification}

Standardized assessment in Colombia is carried out through the Prueba Saber~11, a national examination administered by the Colombian Institute for the Evaluation of Education (ICFES) that every student must take in order to complete secondary education, and whose score largely determines admission to most higher-education institutions \citep{icfes2024saber11}. The test is paired with a socioeconomic questionnaire, providing information that links test scores to students' socioeconomic conditions.

We analyze the second-semester 2023 administration, comprising $n = 415{,}472$ students distributed across $13{,}801$ schools and $1{,}116$ municipalities. The model incorporates the following covariates: \textit{socioeconomic stratum}, an official classification unique to Colombia, assigns each residential area an integer from $1$ (lowest) to $6$ (highest) based on the socioeconomic conditions of the locality. Although stratum is recorded as an integer, we treat it as a categorical covariate with six levels to avoid assuming a linear effect across strata. \textit{Parental education}, which records the highest attainment level of each parent, is coded on an ordered scale: no schooling, primary, secondary, technical, professional, and postgraduate. The remaining covariates are binary: \textit{home internet access} (indicating whether the household has an internet connection), \textit{student employment} (whether the student works while studying), and \textit{school type} (public or private).

The global score $Y_i$ ranges from 0 to 500.
The motivation for analyzing these data is a concern with educational
inequality. Notably, this inequality is most pronounced in the tails of the score distribution rather than at its center: a covariate may increase low scores without affecting high scores, or vice versa. These asymmetries are what a
conditional-mean model is structurally unable to detect, a limitation widely recognized in the educational economics literature \citep{EIDE1998345, levin2001whom}. We model the conditional quantile function at
$\tau \in \{0.10,\,0.25,\,0.50,\,0.75,\,0.90\}$, which traces how each
socioeconomic effect changes across the quantiles.

For each quantile level, we fit the hierarchical
model~\eqref{eq:linear_predictor}, with fixed effects for the
socioeconomic stratum of the household ($1$--$6$), internet access,
student employment, parental educational attainment (mother and father), and school type (public/private). The model has a three-level hierarchy: students are nested within schools, which are in turn nested within municipalities, and we include exchangeable Gaussian
random effects at both the school and the municipality level.
Categorical covariates are coded under sum-to-zero contrasts, so that
every reported effect is a deviation from the grand mean on the
original $0$--$500$ score scale and is directly comparable across
categories and quantile levels. Uncertainty is quantified through the
cluster-robust sandwich correction of Equation~\eqref{eq:sandwich},
whose calibration was established in the simulation study of
Section~\ref{sec:sim_coverage}: the data scores are aggregated at the
municipality level, so that the central matrix preserves the dependence
induced by the spatial clustering of schools, and the prior-precision
term carries the school- and municipality-level random-effect
variability. The latter is relevant here because several covariates,
school type in particular, are constant or strongly clustered within
schools. This yields the $95\%$ intervals reported throughout, computed
from each converged fit in one pass over the data plus a single sparse
solve of dimension $19 + 13{,}801 + 1{,}116$.

\subsection{Computational performance}

The size of the data places the computational demand of this analysis beyond the reach of fully Bayesian sampling. The model involves more than $400{,}000$ observations and close to $15{,}000$ random effects, and is fitted at five quantile levels, so a single HMC fit of a hierarchical model of this size is already infeasible in practice. In contrast, EM-INLA produced the five
point-estimate fits in under $20$ minutes on an Intel Core i5-13600K
with $32$~GB of RAM, and the sandwich intervals of
Equation~\eqref{eq:sandwich} add only a single $O(n)$ pass per
quantile, taking a few seconds each, so the complete
uncertainty-quantified analysis finishes in essentially the time of
the point fits.
The application shows that a full quantile analysis of national data, with interval estimates included, is feasible in minutes.

\subsection{Results}

Table~\ref{tab:sum_to_zero_effects} reports the estimated sum-to-zero
effects at the three most informative quantile levels, together with
the difference $\Delta$ between the estimates at $\tau = 0.90$ and
$\tau = 0.10$, which summarizes how each effect changes from the lower
to the upper tail. The complete set of estimates, covering all
categories and all five quantile levels, is reported in
Appendix~\ref{app:full_coefficients} (Table~\ref{tab:sum_to_zero_full}).
The table shows that almost every covariate
operates with different strength at different points of the score
distribution.

\begin{table}[pos=htbp]
\centering
\caption{Selected sum-to-zero effects on the Saber~11 global score (0--500
scale) at three quantile levels. Brackets show 95\% cluster-robust
sandwich intervals (municipality-level clustering,
Equation~\eqref{eq:sandwich}). $\Delta$ denotes the difference
between the estimates at $\tau = 0.90$ and $\tau = 0.10$.}
\label{tab:sum_to_zero_effects}
\small
\begin{tabular}{@{}lrrrr@{}}
\toprule
\textbf{Effect} & $\bm{\tau = 0.10}$ & $\bm{\tau = 0.50}$ & $\bm{\tau = 0.90}$ & $\bm{\Delta}$ \\
\midrule
Intercept & $201.8\;[200.9,\,202.8]$ & $243.1\;[241.7,\,244.5]$ & $288.5\;[287.1,\,289.9]$ & \\
\addlinespace
Stratum 1 (lowest)        & $8.4\;[8.2,\,8.6]$    & $12.3\;[11.7,\,12.8]$   & $13.2\;[12.8,\,13.6]$   & $+4.8$ \\
Stratum 6 (highest)       & $-9.1\;[-9.7,\,-8.5]$ & $-11.4\;[-12.5,\,-10.2]$ & $-11.4\;[-12.1,\,-10.8]$ & $-2.3$ \\
\addlinespace
Mother: no schooling      & $-7.9\;[-8.3,\,-7.6]$  & $-12.2\;[-13.2,\,-11.1]$ & $-13.9\;[-14.5,\,-13.4]$ & $-6.0$ \\
Mother: postgraduate      & $11.3\;[11.0,\,11.6]$  & $15.7\;[14.1,\,17.2]$    & $16.0\;[15.1,\,16.9]$    & $+4.7$ \\
Father: no schooling      & $-8.8\;[-9.0,\,-8.5]$  & $-12.8\;[-13.6,\,-12.0]$ & $-12.5\;[-12.9,\,-12.1]$ & $-3.8$ \\
Father: postgraduate      & $11.0\;[10.6,\,11.5]$  & $16.5\;[15.3,\,17.7]$    & $14.5\;[13.9,\,15.2]$    & $+3.5$ \\
\addlinespace
Home internet access      & $1.3\;[1.2,\,1.3]$     & $2.6\;[2.4,\,2.8]$       & $3.6\;[3.5,\,3.7]$       & $+2.3$ \\
Works while studying      & $-2.6\;[-2.7,\,-2.5]$  & $-3.6\;[-3.8,\,-3.4]$    & $-3.9\;[-4.0,\,-3.8]$    & $-1.3$ \\
Private school            & $8.2\;[7.7,\,8.6]$     & $8.6\;[8.0,\,9.1]$       & $6.8\;[6.3,\,7.4]$       & $-1.3$ \\
\bottomrule
\end{tabular}
\end{table}

A natural starting point is the intercept path, which is itself
informative about the marginal dispersion of achievement. For a
reference student, the fitted score rises from $201.8$ at
$\tau = 0.10$ to $288.5$ at $\tau = 0.90$, a spread of almost $87$
points on the $0$--$500$ scale. Every socioeconomic effect discussed
below must therefore be read against this wide baseline dispersion:
even the largest covariate effects, those of parental education,
shift the conditional quantiles by considerably less than the gap
that separates low- and high-achieving students with identical
covariates, indicating that a substantial share of the score
variability operates within, rather than between, socioeconomic
profiles.

Parental education exhibits the largest influence. Students with
postgraduate-educated mothers hold an advantage over those whose
mothers have no formal schooling that widens from roughly $19$ points
at the $10$th percentile to nearly $30$ points at the $90$th. The
father's education follows an almost identical trajectory. The
coefficient paths in Figure~\ref{fig:coefficient_paths} make this
pattern clear: the benefits of parental education are
concentrated among already high-achieving students.

\begin{figure}[pos=htbp]
    \centering
    \includegraphics[width=0.8\textwidth]{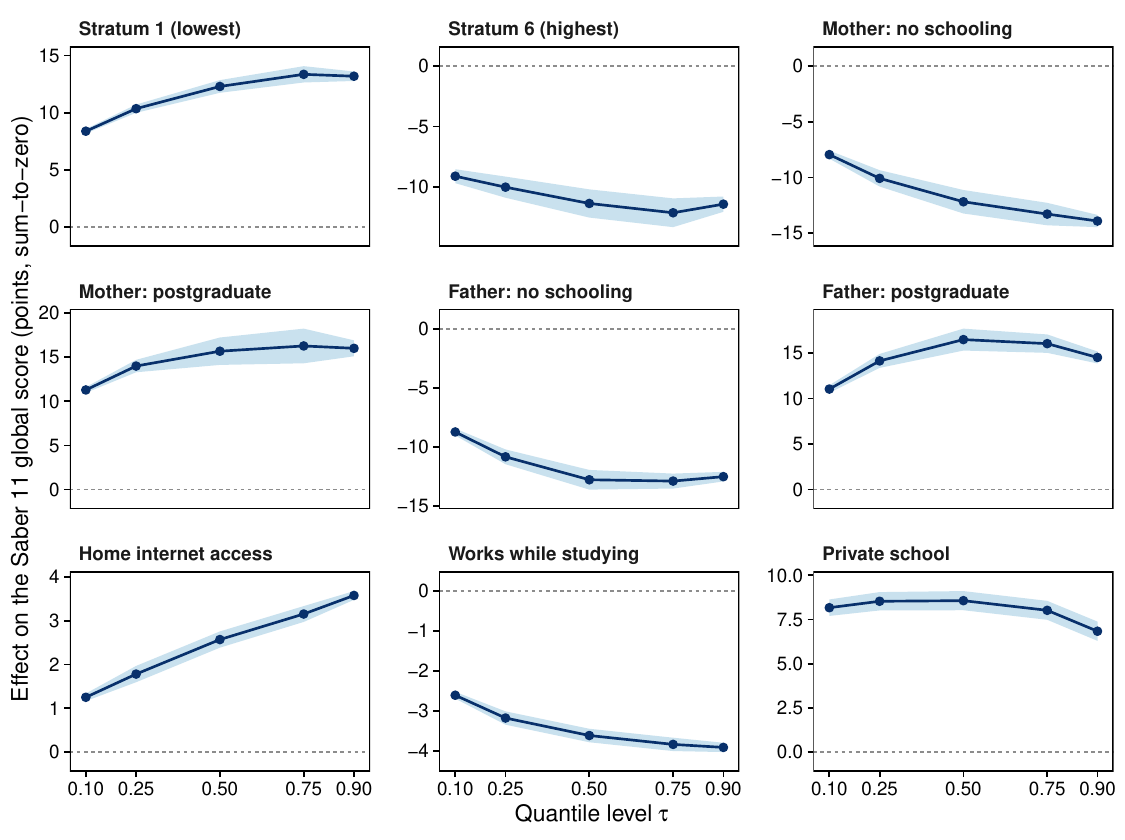}
    \caption{Coefficient paths for selected covariates, with 95\% cluster-robust sandwich bands; the dashed line marks zero.}
    \label{fig:coefficient_paths}
\end{figure}

The remaining covariates display three distinct patterns, summarized
in the clustered forest plot of Figure~\ref{fig:ci_plot}. Internet
access exhibits a widening effect similar to parental education; its
positive impact more than doubles from the lower to the upper tail of
the distribution. The penalty for working while studying remains
stable and persistent across all quantiles. The private-school premium
is largest at the bottom of the distribution and narrows toward the
top, indicating that private schooling primarily benefits
lower-performing students.

\begin{figure}[H]
    \centering
    \includegraphics[width=\textwidth]{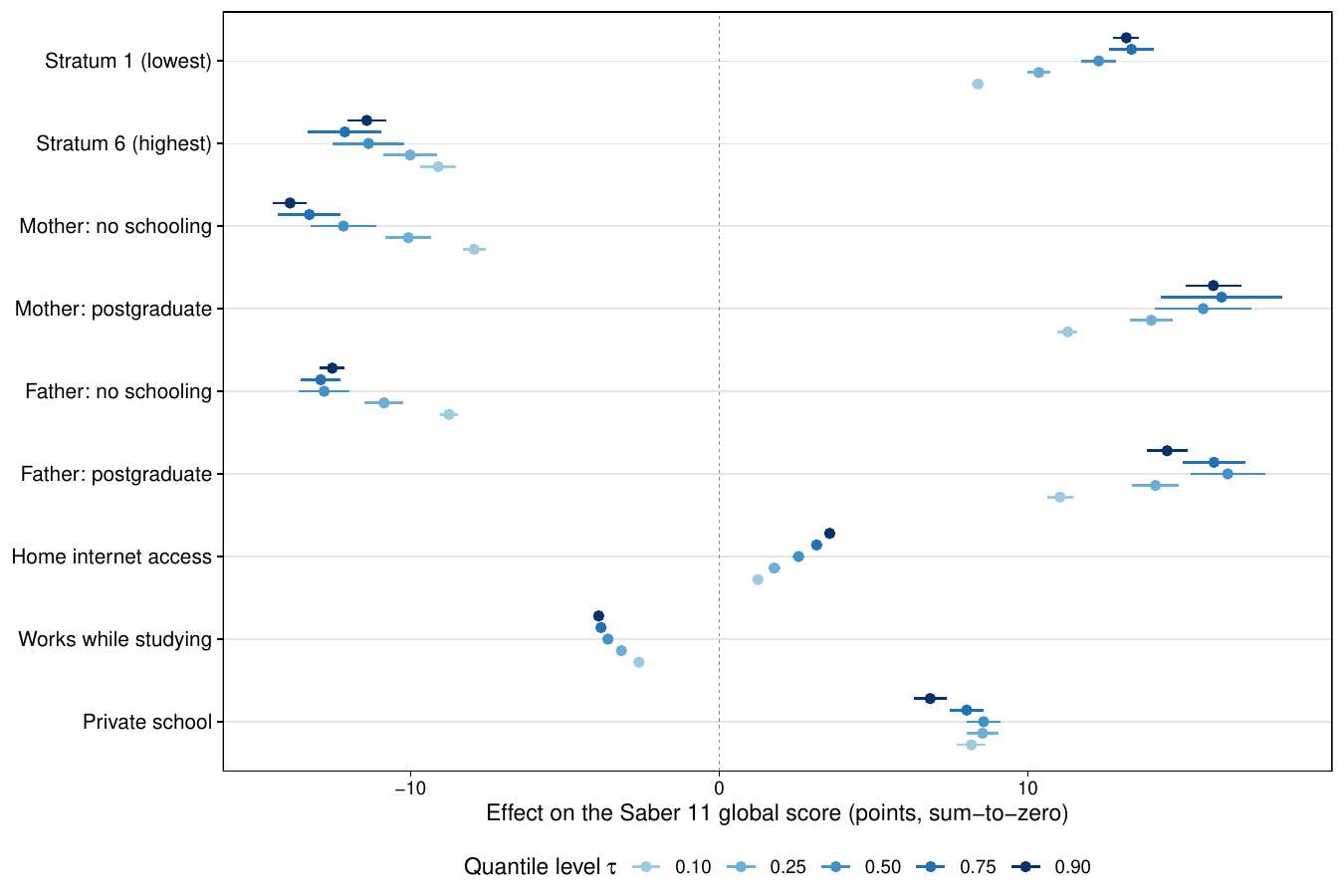}
    \caption{Estimates and 95\% cluster-robust sandwich intervals for selected covariates; darker shades mark higher quantiles.}
    \label{fig:ci_plot}
\end{figure}

The household stratum coefficients should be read with care. While higher strata are unconditionally associated with higher scores, the partial stratum effects in Table~\ref{tab:sum_to_zero_effects} appear to reverse sign. This reversal does not contradict the raw socioeconomic gradient. It reflects high collinearity. Because a higher-stratum household is overwhelmingly more likely to feature high socioeconomic characteristics, such as parents with a high level of education, private schooling, and home internet, the adjusted stratum coefficients capture only the residual variation after these primary channels have been absorbed \citep{gordon1969}. These partial effects should be read strictly as conditional associations, adjusted for the other covariates, rather than as causal effects.

These patterns matter for policy. Because the internet-access advantage is
concentrated in the upper tail, expanding connectivity alone is
unlikely to compress the achievement gap at the bottom of the
distribution. Conversely, the private-school premium, largest among
low-scoring students, suggests that school-level factors matter most
precisely where performance is weakest. A conditional-mean analysis
of the same data would summarize each of these effects by a single
number, averaging away the tail behavior on which such distinctions
rest. Only the quantile decomposition reveals them.
\section{Concluding remarks}
\label{sec:conclusions}

We have proposed EM-INLA, a scalable algorithm for Bayesian hierarchical quantile regression that entirely avoids MCMC sampling. By exploiting the closed-form GIG moments in the E-step and pairing them with analytical M-step updates for the variance parameters, the proposed method inherits the near-linear computational scaling of INLA.
Our simulation study demonstrates that EM-INLA recovers the true parameters with an accuracy comparable to HMC across all tested scenarios and quantile levels, while achieving computational speedups ranging from $15\times$ to over $50\times$. The application to the Prueba Saber~11 dataset (comprising over
$400{,}000$ students nested within schools and municipalities) shows that the algorithm scales to problems where fully Bayesian MCMC approaches become intractable. In this application, the quantile decomposition revealed patterns that would remain hidden under the usual study of the conditional mean, such as the concentration of the internet-access effect in the upper tail and of the private-school premium in the lower tail of the score distribution.

This work presents several contributions. First, it provides a closed-form alternative that avoids smoothing the check loss function, with a runtime comparable to frequentist alternatives while keeping a Bayesian standpoint. Second, since scalability is a usual challenge for Bayesian methods, we showed that the runtime of our method is robust to a high number of random effects and grows sublinearly with the sample size. Third, when compared with frequentist alternatives, the closed-form update we derived for the random-effect variances avoids the boundary collapse observed in scenario M4, both in our simulation study and as reported in the literature \citep{chung2013nondegenerate}.
Despite its computational advantages, the method has theoretical limitations. As with any EM algorithm, convergence to the global optimum is not guaranteed \citep{EM_properties} and the final solution may depend on initialization. However, our proposed warm-start strategy, relying on a preliminary Gaussian fit, proved stable across all evaluated scenarios. Because the exact maximization step of the classical EM algorithm is replaced by an approximate INLA call, the strict monotone ascent property of the objective function is not mathematically guaranteed either. Establishing formal convergence proofs for this hybrid framework remains an important avenue for future theoretical work.
The empirical-Bayes construction inherently understates true posterior uncertainty, as the final INLA call conditions on the converged hyperparameters as if they were fixed and known. Our simulation study makes this gap explicit and motivates the analytic cluster-robust sandwich correction of Equation~\eqref{eq:sandwich}, which restores near-nominal coverage at the cost of a single pass over the data, with the cluster bootstrap (Algorithm~\ref{alg:em_inla_boot}) serving as an independent benchmark that confirms its calibration. Another standard limitation is that each quantile level is fitted independently, so the estimated conditional quantiles are not constrained against crossing, particularly in regions with sparse covariate support.
A natural next step is the enforcement of non-crossing constraints across quantile levels, either through joint estimation \citep{noncrossing} or post-hoc rearrangement \citep{posthoc}, addressing the independent-fitting limitation while preserving computational efficiency. Finally, INLA naturally facilitates the incorporation of spatial random effects for areal and geostatistical data via CAR, BYM2, or SPDE priors \citep{Besag1991, Riebler2016, SPDE}, as well as temporal components through autoregressive or random walk priors. These extensions would expand the applicability of EM-INLA to model complex quantile dynamics in large-scale longitudinal and spatio-temporal settings.
\appendix

\section{Mathematical Derivations}
\label{app:proofs}

This appendix collects the derivations underlying the EM-INLA algorithm,
in the order in which the results appear in Section~\ref{sec:aline}.

\subsection{Conditional GIG distribution and its moments}
\label{app:gig}

Starting from the normal--exponential mixture
representation~\eqref{eq:mixture}, the complete-data likelihood for a
single observation is
\begin{equation*}
    f(y_i, v_i \mid \boldsymbol{\beta}, \sigma)
    = f(y_i \mid v_i, \boldsymbol{\beta}, \sigma)\,f(v_i \mid \sigma),
\end{equation*}
where
$y_i \mid v_i \sim N(\mu_i + \theta v_i,\, \kappa^2\sigma v_i)$,
and $v_i \mid \sigma \sim \mathrm{Exp}(1/\sigma)$.
Setting $r_i^{(t)} = y_i - \mu_i^{(t)}$, the conditional density of
$v_i$ given $y_i$ is proportional to
\begin{align*}
    f(v_i \mid y_i, \boldsymbol{\beta}, \sigma)
    &\propto v_i^{-1/2}
       \exp\!\left\{
         -\frac{(r_i^{(t)} - \theta v_i)^2}{2\kappa^2\sigma^{(t)} v_i}
       \right\}
       \exp\!\left\{-\frac{v_i}{\sigma^{(t)}}\right\} \\
    &= v_i^{-1/2}
       \exp\!\left\{
         -\frac{1}{2}\!\left(
           \frac{\chi_i}{v_i} + \psi\, v_i
         \right)
       \right\},
\end{align*}
where the cross-term from the squared binomial, which is independent of $v_i$, 
has been absorbed into the proportionality constant, and
\begin{equation*}
    \chi_i = \frac{r_i^{(t)2}}{\kappa^2\sigma^{(t)}},
    \qquad
    \psi = \frac{\theta^2}{\kappa^2\sigma^{(t)}} + \frac{2}{\sigma^{(t)}}.
\end{equation*}
This is the kernel of a $\mathrm{GIG}(1/2, \chi_i, \psi)$ distribution,
confirming~\eqref{eq:gig_conditional}. The moments of the general
$\mathrm{GIG}(\lambda, \chi, \psi)$ distribution are
\begin{equation*}
    \mathbb{E}[V]     = \sqrt{\frac{\chi}{\psi}}
                        \frac{K_{\lambda+1}(\sqrt{\chi\psi})}
                             {K_{\lambda}(\sqrt{\chi\psi})},
    \qquad
    \mathbb{E}[V^{-1}]= \sqrt{\frac{\psi}{\chi}}
                        \frac{K_{\lambda-1}(\sqrt{\chi\psi})}
                             {K_{\lambda}(\sqrt{\chi\psi})},
\end{equation*}
where $K_\nu$ denotes the modified Bessel function of the second kind.
For $\lambda = 1/2$ the Bessel functions simplify via
$K_{1/2}(x) = K_{-1/2}(x) = \sqrt{\pi/(2x)}\,e^{-x}$, giving the
closed-form expressions
\begin{equation*}
    \mathbb{E}[v_i^{-1} \mid \cdot] = \sqrt{\frac{\psi}{\chi_i}},
    \qquad
    \mathbb{E}[v_i      \mid \cdot] = \sqrt{\frac{\chi_i}{\psi}}
                                    + \frac{1}{\psi},
\end{equation*}
which are consistent with~\eqref{eq:gig_moments} in the main text.

\subsection{Proof of Theorem~\ref{thm:wls}: weighted Gaussian
            reformulation}
\label{app:wls}

Given current parameter estimates, the $Q$-function with respect to
$(\boldsymbol{\beta}, \boldsymbol{\alpha})$ is
\begin{equation*}
    Q(\boldsymbol{\beta}, \boldsymbol{\alpha})
    = -\frac{1}{2\kappa^2}\sum_{i=1}^{n} \frac{1}{\sigma^{(t)}}
      \mathbb{E}\!\left[
        \frac{(y_i - \mu_i - \theta v_i)^2}{v_i}
      \right] + C,
\end{equation*}
where $\mu_i = \mathbf{x}_i^\top\boldsymbol{\beta}
+ \sum_k\alpha_{g_k(i)}^{(k)}$. Expanding the squared term:
\begin{equation*}
    \mathbb{E}\!\left[
      \frac{(y_i - \mu_i - \theta v_i)^2}{v_i}
    \right]
    = (y_i-\mu_i)^2\,\mathbb{E}[v_i^{-1} \mid \cdot]
       - 2\theta(y_i-\mu_i)
       + \theta^2\,\mathbb{E}[v_i \mid \cdot].
\end{equation*}

\begin{proof}[Proof of Theorem~\ref{thm:wls}]
Setting $r_i = y_i - \mu_i$ and completing the square in $r_i$:
\begin{align*}
    r_i^2\,\mathbb{E}[v_i^{-1} \mid \cdot] - 2\theta\,r_i
    &= \mathbb{E}[v_i^{-1} \mid \cdot]
       \!\left(r_i - \frac{\theta}{\mathbb{E}[v_i^{-1} \mid \cdot]}\right)^{\!2}
       - \frac{\theta^2}{\mathbb{E}[v_i^{-1} \mid \cdot]}.
\end{align*}
Absorbing $-\theta^2/\mathbb{E}[v_i^{-1} \mid \cdot]$ and
$\theta^2\,\mathbb{E}[v_i \mid \cdot]$ into the constant $C$, and
substituting $\tilde{y}_i = y_i - \theta/\mathbb{E}[v_i^{-1} \mid \cdot]$:
\begin{equation*}
    Q \propto -\frac{1}{2}\sum_{i=1}^{n}
      w_i\!\left(\tilde{y}_i - \mu_i\right)^2,
    \qquad w_i = \frac{\mathbb{E}[v_i^{-1} \mid \cdot]}{\sigma^{(t)}\kappa^2}.
\end{equation*}
This is the log-likelihood of $\tilde{y}_i \sim N(\mu_i,\, w_i^{-1})$,
so maximizing $Q$ over $(\boldsymbol{\beta}, \boldsymbol{\alpha})$ is
equivalent to fitting this weighted Gaussian model.
\end{proof}

\subsection{Derivation of the objective function for
            \texorpdfstring{$\sigma$}{sigma}
            and its closed-form maximizer}
\label{app:sigma_derivs}

To update $\sigma$ in the M-step, we isolate the terms in
the complete-data log-likelihood that depend on $\sigma$:
\begin{align*}
    \log f(y_i, v_i \mid \boldsymbol{\beta}, \sigma)
    &= \log f(y_i \mid v_i, \boldsymbol{\beta}, \sigma)
     + \log f(v_i \mid \sigma) \\
    &= -\frac{1}{2}\log(2\pi\kappa^2\sigma v_i)
       - \frac{(r_i^{(t)} - \theta v_i)^2}{2\sigma\kappa^2 v_i}
       - \log\sigma
       - \frac{v_i}{\sigma} \\
    &= C_i
       - \frac{3}{2}\log\sigma
       - \frac{1}{\sigma}
         \left[ \frac{(r_i^{(t)} - \theta v_i)^2}{2\kappa^2 v_i} + v_i \right],
\end{align*}
where $r_i^{(t)} = y_i - \mu_i^{(t)}$ and $C_i$ absorbs all terms
independent of $\sigma$. The coefficient $-3/2$ arises
from combining $-1/2$ from the Gaussian conditional and $-1$ from the
exponential prior.

Taking the expectation over $v_i \mid y_i, \boldsymbol{\beta}^{(t)},
\sigma^{(t)}$ and summing over $i$, we obtain the
objective~\eqref{eq:sigma_objective}:
\begin{equation*}
    Q(\sigma) = -\frac{3}{2}n\log\sigma
    - \frac{1}{\sigma}\sum_{i=1}^{n} A_i + C,
\end{equation*}
where
\begin{align*}
    A_i &= \mathbb{E}\!\left[
             \frac{(r_i^{(t)} - \theta v_i)^2}{2\kappa^2 v_i} + v_i
             \;\middle|\; \cdot
           \right] \\
        &= \frac{1}{\kappa^2}\!\left[
             \frac{1}{2} r_i^{(t)2}\,\mathbb{E}[v_i^{-1} \mid \cdot]
             - \theta\,r_i^{(t)}
             + \left(\frac{\theta^2}{2} + \kappa^2\right)
               \mathbb{E}[v_i \mid \cdot]
           \right].
\end{align*}
Since $A_i$ is the expectation of a strictly positive quantity,
$A_i > 0$ almost surely. Differentiating $Q(\sigma)$ and setting the
result to zero gives the first-order condition
\begin{equation*}
    \frac{\partial Q}{\partial \sigma}
    = -\frac{3n}{2\sigma}
      + \frac{1}{\sigma^2}\sum_{i=1}^{n} A_i = 0
    \implies
    \hat{\sigma}
    = \frac{2}{3n}\sum_{i=1}^{n} A_i,
\end{equation*}
recovering the closed-form update~\eqref{eq:sigma_homo}. The second
derivative evaluated at $\hat{\sigma}$ is $-3n/(2\hat{\sigma}^2) < 0$,
confirming that the unique critical point is a global maximum.
Equivalently, reparametrizing in $\gamma = \log\sigma$ yields
$Q(\gamma) = -\tfrac{3}{2}n\gamma - \bar{A}\,n\,e^{-\gamma}$ with
$\bar{A} = n^{-1}\sum_{i=1}^n A_i$, whose second derivative
$-\bar{A}\,n\,e^{-\gamma} < 0$ holds everywhere, so $Q$ is strictly
concave and admits a unique global maximum.

\subsection{Proof of Proposition~\ref{prop:re}: closed-form update
            for random-effect variances}
\label{app:re}

The contribution of block $k$ to the $Q$-function is
\begin{equation}\label{eq:app_Q_re}
    Q(\sigma_k^2)
    = -\frac{J_k}{2}\log\sigma_k^2
      - \frac{1}{2\sigma_k^2}\sum_{j=1}^{J_k}
        \mathbb{E}\!\left[\alpha_j^{(k)2}
          \mid \tilde{\mathbf{y}}, \hat{\boldsymbol{\beta}}
        \right] + C.
\end{equation}

\begin{proof}[Proof of Proposition~\ref{prop:re}]
Differentiating~\eqref{eq:app_Q_re} with respect to $\sigma_k^2$
and setting the result to zero:
\begin{equation*}
    \frac{\partial Q}{\partial \sigma_k^2}
    = -\frac{J_k}{2\sigma_k^2}
      + \frac{1}{2\sigma_k^4}\sum_{j=1}^{J_k}
        \mathbb{E}\!\left[\alpha_j^{(k)2}
          \mid \tilde{\mathbf{y}}, \hat{\boldsymbol{\beta}}
        \right]
    = 0.
\end{equation*}
Solving for $\sigma_k^2$ and applying the identity
$\mathbb{E}[X^2] = (\mathbb{E}[X])^2 + \mathrm{Var}(X)$:
\begin{equation*}
    \hat{\sigma}_k^2
    = \frac{1}{J_k}\sum_{j=1}^{J_k}
      \Bigl(
        \hat{\alpha}_j^{(k)2}
        + \widehat{\mathrm{Var}}\!\left(\alpha_j^{(k)}\right)
      \Bigr),
\end{equation*}
where $\hat{\alpha}_j^{(k)}$ and
$\widehat{\mathrm{Var}}(\alpha_j^{(k)})$ are the posterior mean and
marginal variance returned by INLA for each random effect. The second
derivative is $-J_k/(2\hat{\sigma}_k^4) < 0$, confirming a maximum.
Since the $K$ blocks enter~\eqref{eq:app_Q_re} additively and
independently, the update is applied in parallel for all $k$.
\end{proof}

\FloatBarrier
\newpage
\subsection{Estimated coefficients}
\label{app:full_coefficients}

Table~\ref{tab:sum_to_zero_full} reports the complete set of
sum-to-zero effects for all categories and all five quantile levels,
complementing the condensed Table~\ref{tab:sum_to_zero_effects} in the
main text.
\begin{table}[pos=htbp]
\centering
\caption{Complete sum-to-zero effects by quantile level. Each cell reports the EM-INLA
point estimate and the 95\% cluster-robust sandwich interval in brackets
(municipality-level clustering, Equation~\eqref{eq:sandwich}). All effects are
deviations from the grand mean on the Saber~11 global score scale (0--500).}
\label{tab:sum_to_zero_full}
\resizebox{\textwidth}{!}{%
\begin{tabular}{llccccc}
\toprule
\textbf{Variable} & \textbf{Category} & \multicolumn{5}{c}{\textbf{Quantile} $\tau$} \\
\cmidrule(lr){3-7}
 & & \textbf{0.10} & \textbf{0.25} & \textbf{0.50} & \textbf{0.75} & \textbf{0.90} \\
\midrule
Intercept & & $201.83\;[200.88,\;202.78]$ & $220.14\;[218.98,\;221.30]$ & $243.10\;[241.70,\;244.50]$ & $266.92\;[265.50,\;268.33]$ & $288.47\;[287.06,\;289.87]$ \\[2pt]
\midrule
\multirow{6}{*}{Stratum}
 & 1 & $8.38\;[8.19,\;8.58]$ & $10.35\;[9.99,\;10.71]$ & $12.29\;[11.74,\;12.85]$ & $13.36\;[12.64,\;14.07]$ & $13.19\;[12.77,\;13.61]$ \\[2pt]
 & 2 & $6.76\;[6.50,\;7.02]$ & $7.81\;[7.40,\;8.23]$ & $8.74\;[8.25,\;9.24]$ & $8.99\;[8.48,\;9.50]$ & $8.58\;[8.28,\;8.88]$ \\[2pt]
 & 3 & $3.27\;[3.04,\;3.50]$ & $3.27\;[2.89,\;3.64]$ & $3.54\;[3.02,\;4.06]$ & $3.27\;[2.88,\;3.66]$ & $2.72\;[2.52,\;2.93]$ \\[2pt]
 & 4 & $-2.73\;[-2.98,\;-2.48]$ & $-3.59\;[-3.96,\;-3.22]$ & $-3.98\;[-4.51,\;-3.44]$ & $-4.72\;[-5.30,\;-4.15]$ & $-4.82\;[-5.16,\;-4.49]$ \\[2pt]
 & 5 & $-6.58\;[-6.98,\;-6.18]$ & $-7.83\;[-8.65,\;-7.01]$ & $-9.24\;[-9.99,\;-8.49]$ & $-8.76\;[-9.37,\;-8.15]$ & $-8.25\;[-8.71,\;-7.79]$ \\[2pt]
 & 6 & $-9.11\;[-9.68,\;-8.53]$ & $-10.01\;[-10.89,\;-9.14]$ & $-11.36\;[-12.53,\;-10.20]$ & $-12.13\;[-13.32,\;-10.94]$ & $-11.42\;[-12.05,\;-10.79]$ \\[2pt]
\midrule
\multirow{2}{*}{Internet access}
 & No  & $-1.25\;[-1.34,\;-1.17]$ & $-1.78\;[-1.97,\;-1.60]$ & $-2.57\;[-2.76,\;-2.39]$ & $-3.15\;[-3.34,\;-2.97]$ & $-3.58\;[-3.68,\;-3.48]$ \\[2pt]
 & Yes & $1.25\;[1.17,\;1.34]$ & $1.78\;[1.60,\;1.97]$ & $2.57\;[2.39,\;2.76]$ & $3.15\;[2.97,\;3.34]$ & $3.58\;[3.48,\;3.68]$ \\[2pt]
\midrule
\multirow{2}{*}{Working hours}
 & Does not work & $2.60\;[2.52,\;2.69]$ & $3.17\;[3.01,\;3.34]$ & $3.61\;[3.44,\;3.78]$ & $3.83\;[3.66,\;4.00]$ & $3.91\;[3.79,\;4.02]$ \\[2pt]
 & Works         & $-2.60\;[-2.69,\;-2.52]$ & $-3.17\;[-3.34,\;-3.01]$ & $-3.61\;[-3.78,\;-3.44]$ & $-3.83\;[-4.00,\;-3.66]$ & $-3.91\;[-4.02,\;-3.79]$ \\[2pt]
\midrule
\multirow{6}{*}{Mother's education}
 & None         & $-7.94\;[-8.30,\;-7.58]$ & $-10.07\;[-10.81,\;-9.34]$ & $-12.18\;[-13.23,\;-11.12]$ & $-13.28\;[-14.29,\;-12.27]$ & $-13.91\;[-14.46,\;-13.36]$ \\[2pt]
 & Primary      & $-5.79\;[-5.93,\;-5.65]$ & $-7.67\;[-7.96,\;-7.38]$ & $-9.11\;[-9.77,\;-8.46]$ & $-9.79\;[-10.60,\;-8.99]$ & $-9.78\;[-10.20,\;-9.35]$ \\[2pt]
 & Secondary    & $-3.53\;[-3.67,\;-3.40]$ & $-4.36\;[-4.59,\;-4.13]$ & $-4.57\;[-5.05,\;-4.08]$ & $-4.33\;[-4.89,\;-3.77]$ & $-3.39\;[-3.68,\;-3.09]$ \\[2pt]
 & Technical    & $3.86\;[3.72,\;4.00]$ & $4.52\;[4.27,\;4.76]$ & $4.88\;[4.56,\;5.21]$ & $4.85\;[4.48,\;5.21]$ & $4.64\;[4.42,\;4.85]$ \\[2pt]
 & Professional & $2.10\;[1.96,\;2.25]$ & $3.59\;[3.28,\;3.91]$ & $5.29\;[4.85,\;5.73]$ & $6.28\;[5.86,\;6.70]$ & $6.42\;[6.14,\;6.71]$ \\[2pt]
 & Postgraduate & $11.29\;[10.97,\;11.61]$ & $14.00\;[13.29,\;14.71]$ & $15.68\;[14.13,\;17.23]$ & $16.28\;[14.31,\;18.24]$ & $16.01\;[15.11,\;16.91]$ \\[2pt]
\midrule
\multirow{6}{*}{Father's education}
 & None         & $-8.76\;[-9.05,\;-8.47]$ & $-10.86\;[-11.49,\;-10.23]$ & $-12.80\;[-13.64,\;-11.97]$ & $-12.91\;[-13.55,\;-12.28]$ & $-12.54\;[-12.94,\;-12.13]$ \\[2pt]
 & Primary      & $-4.28\;[-4.45,\;-4.11]$ & $-5.52\;[-5.82,\;-5.21]$ & $-6.92\;[-7.38,\;-6.47]$ & $-7.58\;[-8.01,\;-7.15]$ & $-7.30\;[-7.62,\;-6.98]$ \\[2pt]
 & Secondary    & $-2.68\;[-2.82,\;-2.53]$ & $-3.42\;[-3.74,\;-3.10]$ & $-4.08\;[-4.61,\;-3.54]$ & $-4.28\;[-4.68,\;-3.89]$ & $-3.91\;[-4.18,\;-3.64]$ \\[2pt]
 & Technical    & $3.10\;[2.93,\;3.27]$ & $3.41\;[3.10,\;3.72]$ & $3.36\;[3.00,\;3.73]$ & $3.56\;[3.20,\;3.92]$ & $3.51\;[3.22,\;3.81]$ \\[2pt]
 & Professional & $1.57\;[1.37,\;1.76]$ & $2.25\;[1.71,\;2.79]$ & $3.96\;[3.14,\;4.79]$ & $5.19\;[4.55,\;5.82]$ & $5.72\;[5.49,\;5.95]$ \\[2pt]
 & Postgraduate & $11.04\;[10.63,\;11.46]$ & $14.13\;[13.37,\;14.90]$ & $16.47\;[15.27,\;17.68]$ & $16.03\;[15.02,\;17.04]$ & $14.51\;[13.86,\;15.17]$ \\[2pt]
\midrule
\multirow{2}{*}{School type}
 & Private & $8.17\;[7.70,\;8.63]$ & $8.53\;[8.01,\;9.05]$ & $8.56\;[8.02,\;9.11]$ & $8.02\;[7.48,\;8.55]$ & $6.83\;[6.29,\;7.38]$ \\[2pt]
 & Public  & $-8.17\;[-8.63,\;-7.70]$ & $-8.53\;[-9.05,\;-8.01]$ & $-8.56\;[-9.11,\;-8.02]$ & $-8.02\;[-8.55,\;-7.48]$ & $-6.83\;[-7.38,\;-6.29]$ \\[2pt]
\bottomrule
\end{tabular}%
}
\end{table}

\begin{figure}[pos=htbp]
    \centering    \includegraphics[width=\textwidth]{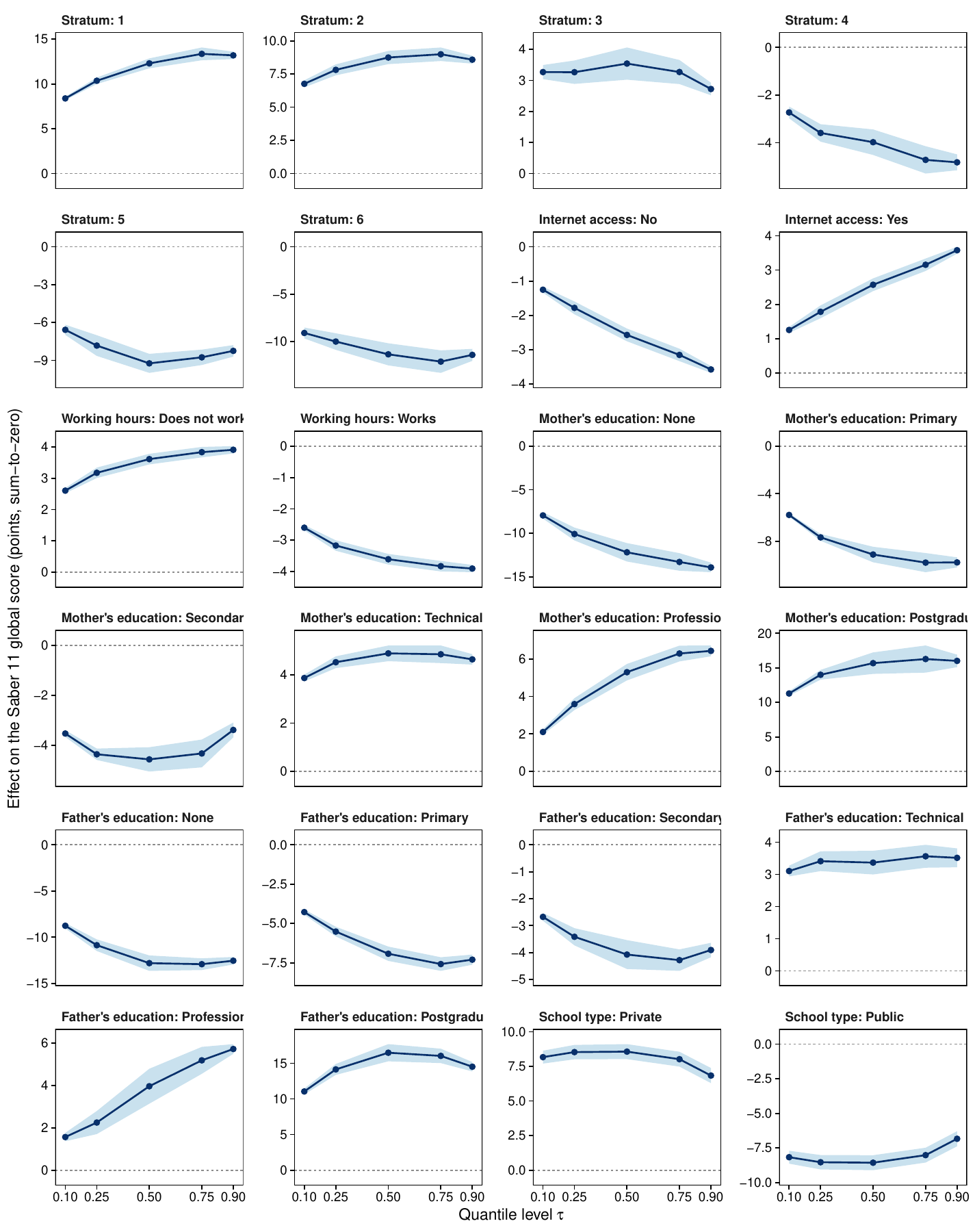}
    \caption{Coefficient paths for all covariates, with 95\% cluster-robust sandwich bands; the dashed line marks zero.}
    \label{fig:coefficient_paths_full}
\end{figure}

\begin{figure}[pos=htbp]
    \centering    \includegraphics[width=\textwidth]{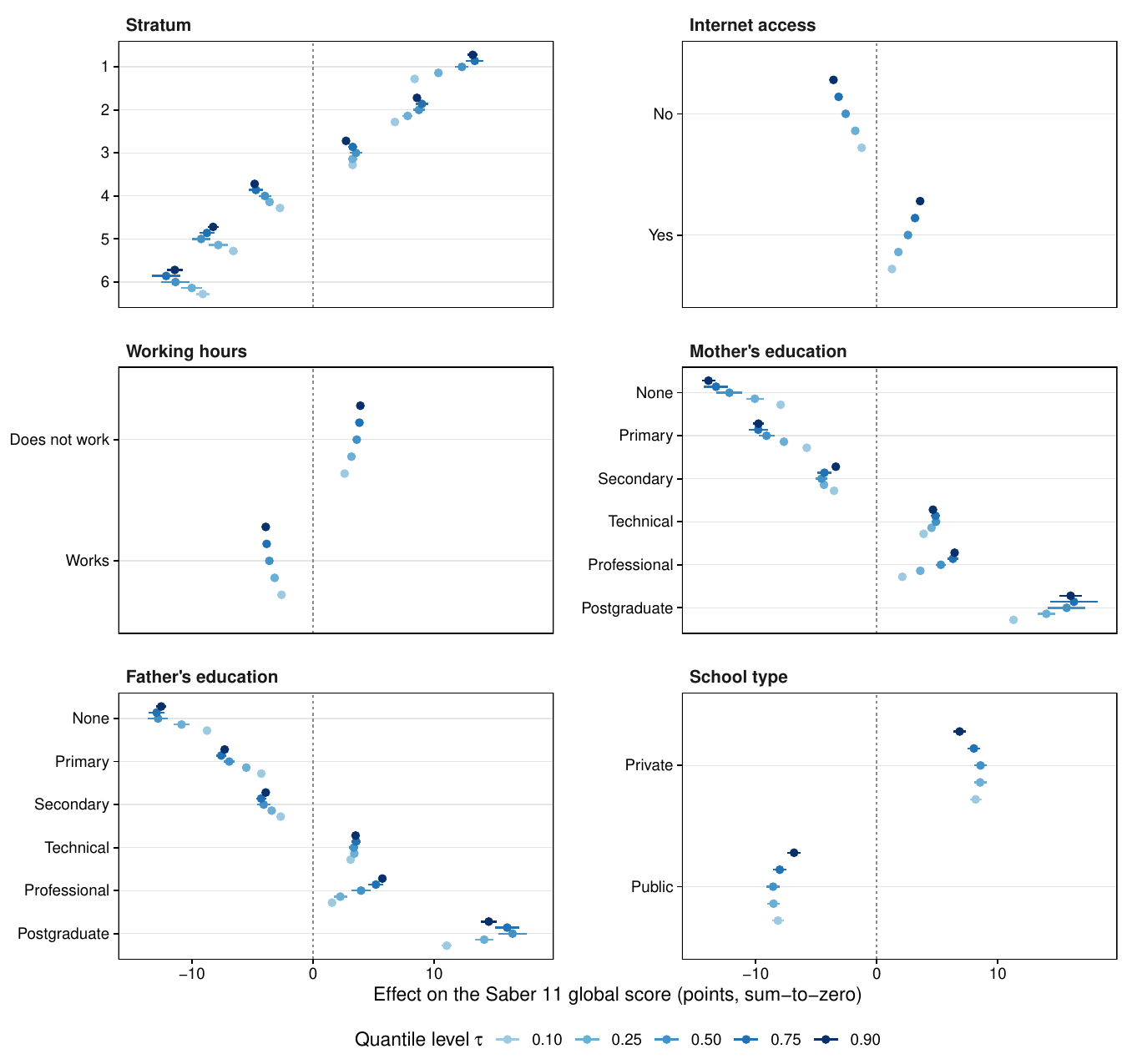}
    \caption{Estimates and 95\% cluster-robust sandwich intervals for all covariates; darker shades mark higher quantiles.}
    \label{fig:forest_plot_full}
\end{figure}

\newpage
\section*{Reproducibility}

All computations were carried out in R~4.4.2. EM-INLA was implemented with the
R-INLA package (version 24.12.11), and the fully Bayesian benchmark with
\texttt{rstan}~2.32.7. The Stan model encodes the ALD working likelihood
through the \texttt{skew\_double\_exponential}$(\mu, 2\sigma, \tau)$
distribution, with a non-centered parameterization for the school-level
effects nested within municipalities. \texttt{adapt\_delta}$=0.90$ and
\texttt{max\_treedepth}$=12$. The sandwich correction of
Equation~\eqref{eq:sandwich} is implemented in base R, using the
\texttt{Matrix} package for the sparse Cholesky solve.

\section*{Data availability}

The Prueba Saber~11 microdata analyzed in
Section~\ref{sec:aplic.data} are publicly available from the Colombian
Institute for the Evaluation of Education (ICFES) through its DataIcfes
portal. Replication code for the simulation study and the application
is available at \href{https://github.com/torodriguezt/EM-INLA}{GitHub}.

\section*{Declaration of competing interest}

The authors declare that they have no known competing financial
interests or personal relationships that could have appeared to
influence the work reported in this paper.

\section*{Funding}

This research did not receive any specific grant from funding agencies
in the public, commercial, or not-for-profit sectors.

\printcredits

\bibliographystyle{cas-model2-names}

\bibliography{references}

\end{document}